\documentclass[12pt,english]{article}
\usepackage{mathptmx}
\usepackage[T1]{fontenc}
\usepackage[latin9]{inputenc}
\usepackage{color}
\usepackage{babel}
\usepackage{enumitem}
\usepackage{amsmath}
\usepackage{amsthm}
\usepackage{amssymb}
\usepackage{geometry}
\usepackage[authoryear]{natbib}
\usepackage[bookmarks=false,
 breaklinks=false,pdfborder={0 0 1},backref=false,colorlinks=true]
 {hyperref}
\hypersetup{
 linkcolor=navy,anchorcolor=blue,citecolor=navy,filecolor=cyan,menucolor=red,runcolor=cyan,urlcolor=navy}

\makeatletter
\theoremstyle{plain}
\newtheorem{lem}{\protect\lemmaname}
\theoremstyle{remark}
\newtheorem{rem}{\protect\remarkname}
\theoremstyle{plain}
\newtheorem{cor}{\protect\corollaryname}
\theoremstyle{plain}
\newtheorem{thm}{\protect\theoremname}
\theoremstyle{remark}
\newtheorem{claim}{\protect\claimname}

\usepackage{babel}
\usepackage{xcolor}
\definecolor{navy}{HTML}{0080AC}

\DeclareMathAlphabet{\mathcal}{OMS}{cmsy}{m}{n}
\DeclareSymbolFont{newfont}{OML}{cmm}{m}{it}
\DeclareMathSymbol{\epsilon}{3}{newfont}{15}
\DeclareMathSymbol{\varrho}{3}{newfont}{37}

\makeatletter
\renewcommand{\thebibliography}[1]{%
  \section*{\refname}%
    \@mkboth{\MakeUppercase\refname}{\MakeUppercase\refname}%
    \list{\@biblabel{\@arabic\c@enumiv}}%
         {\settowidth\labelwidth{\@biblabel{#1}}%
          \leftmargin\labelwidth
          \advance\leftmargin\labelsep
          \@openbib@code
          \usecounter{enumiv}%
          \let\p@enumiv\@empty
          \renewcommand\theenumiv{\@arabic\c@enumiv}}%
    \sloppy
    \clubpenalty4000
    \@clubpenalty \clubpenalty
    \widowpenalty4000%
    \sfcode`\.\@m
    \linespread{1}\selectfont 
}

\makeatother

\usepackage{footmisc}
\makeatother

\providecommand{\claimname}{Claim}
\providecommand{\corollaryname}{Corollary}
\providecommand{\lemmaname}{Lemma}
\providecommand{\remarkname}{Remark}
\providecommand{\theoremname}{Theorem}

\begin{document}
\title{Peace Through Side Payments\thanks{We thank Subhasish M. Chowdhury, Shiran Rachmilevitch and conference
participants at 2021 Contests: Theory and Evidence Conference at Max
Planck Institute, 2021 Annual SAET Conference at Seoul National University
for helpful comments and suggestions.}}
\author{Jingfeng Lu\thanks{Department of Economics, National University of Singapore. Email:
\protect\protect\protect\href{mailto:ecsljf@nus.edu.sg}{ecsljf@nus.edu.sg}.} $\quad$ Zongwei Lu\thanks{School of Economics, Shandong University, \protect\protect\protect\href{mailto:zongwei.lu@email.sdu.edu.cn}{zongwei.lu@email.sdu.edu.cn}.}
$\quad$ Christian Riis\thanks{Department of Economics, BI Norwegian Business School, \protect\protect\protect\href{mailto:christian.riis@bi.no}{christian.riis@bi.no}.}}
\maketitle
\begin{abstract}
We study strategic bargaining for peaceful settlement before a conflict
escalates into war, comparing two protocols: offering a take-it-or-leave-it
side payment versus requesting one. Unlike in mediation models, the
proposer here can signal private information and influence the opponent's
beliefs by varying the payment proposal, which renders the prospect
of peace tenuous. In the bribing model, peace can be implemented through
a continuum of bribes but cannot be secured. Conversely, in the requesting
model, peace security is possible, yet it can only be sustained through
a single, specific request.

\bigskip{}
 \bigskip{}


\bigskip{}

\end{abstract}
\clearpage{}

\section{Introduction}

Before a conflict escalates, disputants may attempt to settle via
a side payment. Such payments are common in legal disputes, political
conflicts, and international bargaining, where one side compensates
the other for staying out of a costly contest. Under what conditions
does such a side payment lead to or even guarantee peace? To answer
this question, we study two bargaining protocols with opposite payment
directions. In the first (bribing) protocol, before a single-object
all-pay auction starts, player 1 can offer a take-it-or-leave-it side
payment (also called a bribe) to player 2. If player 2 accepts, she
exits the auction and player 1 obtains the object without costly conflict.
If player 2 rejects, both players enter the all-pay auction, escalating
the conflict. In the second (requesting) protocol, instead of offering
a side payment, player 1 requests one from player 2 for his own exit
from the conflict. 

Our models admit several degrees of peaceful settlement. We adopt
and adapt the notions of peace prospects introduced by \citet{Zheng2019}.
If peace occurs with certainty in a perfect Bayesian equilibrium (hereafter,
equilibrium), we say that the equilibrium is \emph{peaceful}. If there
exists such a peaceful equilibrium (supported by a belief system),
then we say that peace is \emph{implementable}. Furthermore, if an
equilibrium survives the D1 criterion of \citet{cho1990}, we say
that it is \emph{robust}.

Motivated by the idea of self-confirming equilibria (e.g., \citet{Fudenberg_Levine1993}),
\citet{Zheng2019} develops a notion of peace security to address
potential opportunism by requiring that peace remains optimal under
any conceivable continuation play. While we build upon the framework
of \citet{Zheng2019}, the strategic structure of our model necessitates
a refinement of the notion of peace security. In \citet{Zheng2019},
a deviation by either player leads immediately to a continuation auction;
thus, security primarily concerns whether a player can exploit the
opponent's off-path beliefs within that auction. In our setting, however,
the proposer can deviate to various off-path bribes or requests. If
security were defined relative to arbitrary replies, it would be trivially
impossible for any positive bribe: for instance, the proposer could
deviate to any lower off-path bribe, and a reply of full acceptance
(though perhaps irrational) would render the deviation profitable.\footnote{In the requesting model, peace security would be trivially impossible:
the proposer could deviate to a marginally higher request, and a
reply of full acceptance (though perhaps irrational) would render
the deviation profitable.} To avoid this, we define security relative to \emph{consistent} continuation
replies. A continuation reply is consistent if the receiver's decision
at the earlier stage (acceptance or rejection) is consistent with
the expected payoff from the subsequent auction equilibrium. This
adaptation ensures that we evaluate the on-path peaceful settlement
only against deviation payoffs that are actually conceivable given
the receiver's rationality. Crucially, this definition still captures
the core challenge of security: once a consistent reply is fixed,
a deviating proposer may still ``secretly'' bid in the continuation
auction, using his private knowledge of his own type to exploit the
receiver's beliefs.

We show that peace security is impossible in our bribing model, in
sharp contrast to \citet{Zheng2019}. The key difference lies in the
proposer's incentives. In \citet{Zheng2019}, a neutral mediator can
allocate surplus evenly to sustain security. In our model, player
1 is a self-interested proposer who retains all surplus. The ability
to propose different bribes creates a persistent incentive for player
1 to undercut the on-path bribe and capture additional rent. Because
player 1's highest type always earns positive information rent in
the continuation auction, this undercutting deviation is always profitable,
rendering peace security impossible.

Peace implementability, however, is possible in the bribing model.
We characterize the necessary and sufficient conditions for implementability.
Under these conditions, there is generally a continuum of equilibria;
all are robust, but they yield different payoffs to the players. 
Peace implementability depends critically on the incentives of player
1's lowest type; specifically, the off-path ``gambling'' for this
type due to the ability to propose various bribes can be interpreted
as participation in a particular first-price auction (FPA).

We then analyze the requesting model, where player 1 requests a side
payment for his own exit. We characterize the necessary and sufficient
conditions for robust peaceful equilibria. In contrast to the bribing
model, all robust peaceful equilibria feature a unique on-path request:
player 2's lowest valuation (the highest possible peaceful request).
The intuition follows from the D1 criterion: any off-path request
is attributed to player 1's highest type. So any request lower than
player 2's lowest valuation is accepted by all types of player 2.
But player 1 would not allow such a surplus margin to remain for player
2. Hence, the only possible equilibrium request is the lowest valuation
of player 2.

In contrast to the bribing model, peace security is possible in the
requesting model. Here, all types of player 1 receive the maximum
peaceful rent, eliminating the undercutting incentive found in the
bribing case. Security depends on two constraints. First, player 2's
lowest type must not profit from rejection, even if player 1 holds
the most favorable belief that player 2 is strong. Second, player
1's strongest type must not profit from raising the request to trigger
full rejection and exploit the continuation auction. Under player
2's most optimistic belief that player 1 is weak, the auction features
low bids, allowing the strong type to win cheaply. Security requires
this potential gain (the gap between player 1's highest and lowest
valuations) to be bounded by player 2's lowest valuation (the on-path
payoff). If player 1's lowest type exceeds player 2's lowest type,
the second constraint is automatically satisfied.

To summarize, in the bribing model, peace can be implemented through
a continuum of bribes but cannot be secured. Conversely, in the requesting
model, peace security is possible, but it can be sustained only through
a single, specific request. Collectively, these results highlight
the critical role of the neutral mediator in \citet{Zheng2019}; while
a mediator can secure peace through a continuum of proposals (under
suitable conditions), the introduction of a proposer who can signal
private information renders the prospect of peace far more tenuous.

In addition to the main results, for the bribing model we also examine
the comparative statics of peace implementability. In \citet{Zheng2019},
strengthening a player in the sense of first-order stochastic dominance
(FOSD) preserves peace implementability. This conclusion continues
to hold for player 1 in our model. However, for player 2, FOSD is
insufficient because player 1's off-path payoffs depend on truncated
conditional distributions. We show that peace implementability is
preserved if FOSD is strengthened to hazard-rate dominance, which
maintains the dominance order under truncation.

Furthermore, we also examine non-peaceful equilibria of the bribing
model. We rule out regular separating equilibria, contrasting with
\citet{ES2004} but echoing \citet{rachmilevitch2013}. The intuition
is that a higher type always prefers to mimic a marginally lower type
by offering a lower bribe. This deviation yields two benefits: it
saves money if accepted, and lowers the continuation auction's competitive
intensity if rejected (allowing the higher type to win with certainty).
These gains outweigh the cost of competing against slightly more types,
precluding any continuous, strictly increasing bribing function. We
also characterize equilibria with finite pooling regions. Such equilibria
feature either one or two distinct bribes. If there are two, the lower
bribe is rejected with positive probability, while the higher bribe
is accepted with certainty.

The rest of the article is organized as follows. A brief review of
the related literature appears in the remainder of this section. In
Section~\ref{sec:The-model}, we describe the model and introduce
some useful results for all-pay auctions from \citet{Zheng2019}.
In Section \ref{sec:The-bribing-game}, we first show that peace security
is impossible and then characterize the necessary and sufficient conditions
for peace implementability. We then consider non-peaceful equilibria.
In Section~\ref{sec:requesting}, we characterize the necessary and
sufficient conditions for robust peaceful equilibria and for peace
security. Section~\ref{sec:Conclusion} concludes. The additional
results are gathered in the supplemental appendix. 

\subsection{Related literature}

Our work is most closely related to \citet{ES2004} and \citet{Zheng2019}.

\citet{ES2004} pioneer the literature on dynamic models for bidder
collusion. They consider second-price auctions and study the notion
of bribe-proofness, namely the possibility that a peaceful settlement
cannot be achieved (in some equilibria). They show that under certain
regularity conditions, robust increasing equilibria exist. The difference
is that we consider all-pay auctions and we are interested in the
possibility of peaceful settlement. The nature of all-pay auctions
complicates the analysis because players have no (weakly) dominant
strategies in the continuation auctions if player 2 rejects a bribe.
On the other hand, thanks to \citet{Zheng2019}, our analysis is made
possible by his results on two-player all-pay auctions (with arbitrary,
independent type distributions). Since \citet{ES2004}, subsequent
work has proposed alternative models for second-price and first-price
auctions. \citet{rachmilevitch2013} considers the same bargaining
protocol in first-price auctions, while \citet{Rach2015} allows for
alternating offers between two players in second-price auctions. \citet{Troyan2017}
extends the model of \citet{ES2004} to a setting with interdependent
valuations and affiliated signals. Our earlier article, \citet{LU20211},
also examines second-price auctions but allows for a combination of
a bribe and a request. The key difference between the current paper
and this strand of literature is that we consider all-pay auctions
and focus on peaceful settlement.

\citet{Zheng2019} considers a conflict-mediation model in which,
before an all-pay auction starts, a mediator can recommend a split
of the prize. If both players accept the split, the game ends and
peace is achieved; otherwise, both players enter the auction and compete
non-cooperatively. The key difference is that in our model one player
is endowed with bargaining power and bargains directly with the opponent.
One consequence is that, in our model, the briber can propose different
off-path bribes, which endogenize player 2's beliefs and replies.
In particular, in the continuation games of our model, both players'
type distributions may be updated, whereas in Zheng's model only the
rejector's type distribution is updated. 

More generally, our work relates to the conflict mediation literature.
One strand assumes an exogenous outside option of peace, with a mediator
choosing a mechanism to preempt conflict (e.g., \citet{Bester2006, compte2009, fey2011, horner2015, spier1994pretrial}).
A second strand, to which our paper contributes, endogenizes the outside
option as the continuation play if mediation fails. Within this strand,
\citet{Zheng2019b} considers the same mediation model for first-price
auctions. \citet{Balzer2021a} consider a particular mechanism---Alternative
Dispute Resolution (ADR). Assuming discrete type distributions, they
characterize the optimal ADR mechanism that maximizes the settlement
rate. \citet{zhengali2021} revisit the all-pay auction model and
explore the case where peace cannot be guaranteed. In their paper,
the mediator's objective is to maximize the sum of the ex ante payoffs
of the players.\footnote{A separate strand studies conflict avoidance under complete information,
often modeling wars as Tullock contests (e.g., \citet{bevia2010, Herbst2017, kimbrough2013}).}

Finally, our work contributes to the literature on auction collusion.
In the classic framework (e.g., \citet{Graham1987}), collusion is
often modeled as a static game where a pre-existing cartel coordinates
bidding to extract surplus from the auctioneer. One strand of subsequent
research has extended this framework to dynamic settings but with
repeated auctions, where coordination is achieved through explicit
side payments or communication (e.g., \citet{Aoyagi2007}), intertemporal
payoff transfers via bid rotation (e.g., \citet{Aoyagi2003}), or
tacit endogenous bid rotation without communication or transfers (e.g.,
\citet{Rach2013b}). While these studies focus on the efficiency and
feasibility of coordinating actions in repeated settings, our paper
examines a different strategic question: the prospects for peaceful
settlement. In particular, in our setting, bidders bargain directly
in a one-shot interaction, and the proposed side payment can serve
as a signaling device.\footnote{\citet{horner2007} show that bidding  can be used as a signal to
manipulate the intensity of competition through bluffing or sandbagging,
thereby making non-monotonic signaling possible when jump bidding
is allowed. }

\section{The models \protect}\label{sec:The-model}

\global\long\def\supp{\mathrm{{supp}}}%
In this section, we describe the formal setup of the bribing model.
Since the requesting model shares the same underlying environment
and information structure, differing only in the direction of the
payment, the description is omitted here for brevity.

Two risk-neutral players, player 1 (he) and player 2 (she), are about
to participate in an all-pay auction. Before the auction begins, player
1 can offer a take-it-or-leave-it bribe $b$ to player 2. If player
2 accepts the bribe, she exits the auction and player 1 wins the object
at zero cost; otherwise, both players enter the auction and compete
non-cooperatively. The standard all-pay auction format is assumed;
in particular, ties are broken by a fair coin flip.

For each $i\in\{1,2\}$, player $i$'s valuation $v_{i}$ (also referred
to as the player's type) for the object is independently distributed
according to an atomless continuous distribution $F_{i}$ with density
$f_{i}$ on the support $[\underline{v}_{i},\bar{v}_{i}]$, where
$\bar{v}_{i}>\underline{v}_{i}\ge0$. 

We focus on perfect Bayesian equilibria. In the bargaining stage,
player 1 uses a pure strategy that maps each type to a bribe, and
player 2 uses a pure acceptance-rejection strategy that prescribes
either acceptance or rejection for each type and each observed bribe.
In the continuation auctions, each type of each player is allowed
to use a mixed bidding strategy because of the nature of all-pay auctions;
we work with distributional strategies in the same sense as \citet{Zheng2019}.

Upon receiving an off-path bribe $b$, player 2 forms a belief $\tilde{F}_{1}$
about player 1's type and considers an acceptance-rejection function
$\varrho$ that specifies which types accept and which reject. The
induced type distribution of player 2 in the continuation auction
is denoted by $\tilde{F}_{2}$. We assume that $\tilde{F}_{1}$ and
$\tilde{F}_{2}$ are independent. Player $i$'s distributional strategy
in the continuation auction is a probability measure on the product
space of types and bids. We use $\sigma$ to denote the equilibrium
(strategy pair) of a continuation auction $\mathcal{G}(\tilde{F}_{1},\tilde{F}_{2})$.
To distinguish $\sigma$ from the equilibrium of the grand game, we
refer to $\sigma$ as a BNE (Bayesian Nash equilibrium). For convenience
of exposition, we use $\pi_{i}$ to denote the payoffs from the grand
game and $U_{i}$ to denote the payoffs from the continuation auctions.
Finally, $\delta_{x}$ denotes the Dirac probability measure supported
on $\{x\}$. 

Player $i$'s strategy in a BNE $\sigma$ yields a bid distribution
$H_{i,\sigma}(\cdot)$ and $c_{i,\sigma}:=H_{i,\sigma}(0)$. Thus,
against player $i$, a bid $\beta$ wins with probability $H_{i,\sigma}(\beta)$.
Because of the payment rule of an all-pay auction and the equilibrium
conditions, for any $\sigma$, the supports of two bid distributions
coincide; we denote the highest possible bid by $x_{\sigma}$.

A strategy-belief assessment in the grand game specifies player 1's
bribe strategy, player 2's acceptance-rejection strategy after every
bribe, player 2's belief $\tilde{F}_{1}$ after every off-path bribe,
and a BNE $\sigma$ of the continuation auction induced after every
rejection history. On the equilibrium path, beliefs follow Bayes's
rule. Off the equilibrium path, we impose the following consistency
requirement (to avoid self-contradiction in players' replies): $\varrho$
and $\mathcal{G}(\tilde{F}_{1},\tilde{F}_{2})$ must be consistent
with each other:\footnote{Although we describe the requirements for the bribing model, the corresponding
requirements apply to the requesting model in Section \ref{sec:requesting}.} 
\begin{enumerate}[label=(\roman{enumi}).]
\item The acceptance-rejection behavior of player 2's types in $\varrho$
is consistent with the equilibrium outcome in $\mathcal{G}(\tilde{F}_{1},\tilde{F}_{2})$
(and belief $\tilde{F}_{1}$); if a type $v_{2}$ accepts an off-path
bribe $b$, she should not find it profitable to submit any bid in
the continuation auction; if a type $v_{2}$ rejects an off-path bribe
$b$, then her payoff from the continuation auction should not be
lower than $b$. In particular, when the reply is full acceptance,
no type of player 2 should find it profitable to reject the bribe
and individually compete with player 1 under belief $\tilde{F}_{1}$.
\item $\tilde{F}_{2}$ is consistent with $\varrho$ and the equilibrium
outcome of $\mathcal{G}(\tilde{F}_{1},\tilde{F}_{2})$, and is common
to both players.\footnote{Given the rejection set implied by $\varrho$, Bayes's rule is applicable.} 
\end{enumerate}
Off-path consistency does not rule out the following possibility.
For an off-path bribe $b$ and a belief $\tilde{F}_{1}$, there may
exist multiple pairs of $\varrho$ and $\sigma$ that satisfy the
requirements; we call these consistent replies. 

Off-path consistency implies that if, for an off-path bribe, the lowest
rejecting type is in the interior of the type support, then this type
must be indifferent between accepting the bribe and rejecting it.
Hence, the indifference condition uniquely determines the candidate
threshold.

Finally, for the continuation auction $\mathcal{G}(F_{1},\tilde{F}_{2})$
following the rejection of an on-path pooling bribe, we use $\sigma_{2}$
to denote an associated BNE.

\subsection{Preliminaries}

In this section, we introduce some properties of equilibria in two-player
all-pay auctions (with independent type distributions) from \citet{Zheng2019}.
Throughout the paper, we follow Zheng's definition of a generalized
inverse for the distribution function $\tilde{F}_{i}$: 
\[
\tilde{F}_{i}^{-1}(x):=\inf\{t\in\supp\;\tilde{F}_{i}:\tilde{F}_{i}(t)\ge x\}.
\]

\begin{lem}
\label{general-all-pay} For any all-pay auction $\mathcal{G}(\tilde{F}_{1},\tilde{F}_{2})$
and any BNE $\sigma$ of $\mathcal{G}(\tilde{F}_{1},\tilde{F}_{2})$,
there exists a unique triple $(x_{\sigma},c_{1,\sigma},c_{2,\sigma})\in\mathbb{R}_{++}\times[0,1]^{2}$
such that for each $i\in\{1,2\}$, $H_{i,\sigma}(x_{\sigma})=1$ and
for all $\beta\in[0,x_{\sigma}]$, 
\[
H_{i,\sigma}(\beta)=c_{i,\sigma}+\int_{0}^{\beta}\frac{1}{\tilde{F}_{-i}^{-1}(H_{-i,\sigma}(y))}dy,
\]
where $c_{i,\sigma}:=H_{i,\sigma}(0)$. 

Furthermore, 
\begin{enumerate}[label=(\roman{enumi}).]
\item $c_{1,\sigma}c_{2,\sigma}=0$; 
\item given the equilibrium behavior of types $v_{-i}\in\supp\;\tilde{F}_{-i}$
in $\sigma$, for any type $v_{i}$, the supremum expected payoff,
denoted by $U_{i}(v_{i}|\sigma)$, is 
\[
U_{i}(v_{i}|\sigma)=\max_{\beta\in\mathbb{R}_{+}}\;H_{-i,\sigma}(\beta)v_{i}-\beta.
\]
\end{enumerate}
\end{lem}
\begin{proof}
See Theorems 5 and 6 in \citet{Zheng2019}. 
\end{proof}
\begin{rem}
\label{rem:A-useful-interpretation}A useful interpretation of $c_{1,\sigma}c_{2,\sigma}=0$
is that the lowest types of $\tilde{F}_{1},\tilde{F}_{2}$ (denoted
by, say, $\underline{v}_{1},\underline{v}_{2}$) cannot both earn
positive expected payoffs because 
\[
U_{i}(\underline{v}_{i}|\sigma)=c_{-i,\sigma}\underline{v}_{i}.
\]
It also implies that the highest bid $x_{\sigma}$ must not be lower
than $\min\{\underline{v}_{1},\underline{v}_{2}\}$. Moreover, according
to Lemma 6 in \citet{Zheng2019}, the strategy of each player is monotone.
Thus, for the highest type of each player (denoted by, say, $\bar{v}_{i}$),
bidding $x_{\sigma}$ is a best response and the expected payoff is
\[
U_{i}(\bar{v}_{i}|\sigma)=\bar{v}_{i}-x_{\sigma}.
\]
\end{rem}

Let 
\[
\mathcal{E}_{i}(\tilde{F}_{i}):=\text{ the set of BNEs of }\ensuremath{\mathcal{G}(\tilde{F}_{i},\tilde{F}_{-i}).}
\]
The following lemma follows from Lemma \ref{general-all-pay} and
concerns the equilibria of one-sided complete information all-pay
auctions. 
\begin{lem}
\label{lem:one-sided-complete-info}For any $i\in\{1,2\}$, and any
$v_{i}^{*}\ne0$, $\mathcal{E}_{i}(\delta_{v_{i}^{*}})=\{\sigma^{*}\}$
such that 
\begin{align}
\forall\beta\in[0,x_{\sigma^{*}}]:H_{i,\sigma^{*}}(\beta) & =c_{i,\sigma^{*}}+\int_{0}^{\beta}\frac{1}{\tilde{F}_{-i}^{-1}\left(\frac{s}{v_{i}^{*}}+c_{-i,\sigma^{*}}\right)}ds;\label{eq:degenerate-boundary-cond}\\
\forall\beta\in[0,x_{\sigma^{*}}]:H_{-i,\sigma^{*}}(\beta) & =\frac{\beta}{v_{i}^{*}}+c_{-i,\sigma^{*}};\nonumber \\
c_{i,\sigma^{*}}c_{-i,\sigma^{*}} & =0;\nonumber \\
x_{\sigma^{*}} & =v_{i}^{*}(1-c_{-i,\sigma^{*}}).\nonumber 
\end{align}
\end{lem}
\begin{proof}
See Lemma 13 in \citet{Zheng2019}. 
\end{proof}
\begin{rem}
\label{rem:boundary-cond-one-sided}The boundary condition $H_{i,\sigma^{*}}(x_{\sigma^{*}})=1$
implies 
\begin{equation}
c_{i,\sigma^{*}}=1-v_{i}^{*}\int_{c_{-i,\sigma^{*}}}^{1}\frac{1}{\tilde{F}_{-i}^{-1}(s)}ds.\label{eq:boundary-cond-one-sided}
\end{equation}
It follows from the above that: 
\begin{enumerate}[label=(\roman{enumi}).]
\item If $c_{i,\sigma^{*}}>0$, then $c_{-i,\sigma^{*}}=0$. This in turn
implies that type $v_{i}^{*}$'s expected payoff $U_{i}(v_{i}^{*}|\sigma^{*})=0$
and $x_{\sigma^{*}}=v_{i}^{*}$. 
\item $x_{\sigma^{*}}<v_{i}^{*}$ is equivalent to $c_{-i,\sigma^{*}}>0$.
To see this, suppose first that $x_{\sigma^{*}}<v_{i}^{*}$. Then
it follows from above that $c_{-i,\sigma^{*}}>0$. On the other hand,
if $c_{-i,\sigma^{*}}>0$, then the probability of player $-i$ bidding
zero is strictly positive in $\sigma^{*}$,\footnote{It means that the lowest possible type must earn zero payoff in $\sigma^{*}$.}
which implies $x_{\sigma^{*}}>\inf\;\supp\;\tilde{F}_{-i}$ and 
\[
U_{i}(v_{i}^{*}|\sigma^{*})=v_{i}^{*}-x_{\sigma^{*}}=v_{i}^{*}c_{-i,\sigma^{*}}>0.
\]
It follows that $x_{\sigma^{*}}<v_{i}^{*}$. 
\end{enumerate}
\end{rem}
The following lemma concerns the expected payoff of type $v_{i}$
in an all-pay auction where player $-i$ believes that the common
belief of the game is $\mathcal{G}(v_{i}^{*},\tilde{F}_{-i})$, while
player $i$ secretly knows that his own type is $v_{i}$. 
\begin{lem}
\label{secret-bidding} Given the above $H_{-i,\sigma^{*}}$, for
any type $v_{i}$, if $v_{i}\le v_{i}^{*}$, then bidding zero is
a best response to $H_{-i,\sigma^{*}}$, and the expected payoff is
$c_{-i,\sigma^{*}}v_{i}$; if $v_{i}>v_{i}^{*}$, then bidding $x_{\sigma^{*}}$
is the best response, yielding an expected payoff of $v_{i}-x_{\sigma^{*}}$
(equivalently, $v_{i}-v_{i}^{*}(1-c_{-i,\sigma^{*}})$). 
\end{lem}
\begin{rem}
Technically, as explained in Lemma 14 and equation (16) in \citet{Zheng2019},
if $v_{i}<v_{i}^{*}$ and $c_{-i,\sigma^{*}}>0$, then the best reply
to $H_{-i,\sigma^{*}}$ is null, but $\lim_{\epsilon\downarrow0}\sup\;\mathrm{BR}_{i}(v_{i},\epsilon|\sigma^{*})=0$
($\mathrm{BR}$ stands for best response). Since the supremum expected
payoff of type $v_{i}$ from bidding positive bids is $c_{-i,\sigma^{*}}v_{i}$,
we simplify the exposition by stating that bidding zero is a best
response to $H_{-i,\sigma^{*}}$ as if the tie-breaking rule is altered
in this particular case to assigning a probability one of winning
to type $v_{i}$ (against those types of player $-i$ who also bid
zero according to $H_{-i,\sigma^{*}}$). For the proof of Lemma \ref{secret-bidding},
see Lemma 14 in \citet{Zheng2019}. 

The following lemma states an intuitive result about the positivity
of the expected payoff of each player's highest type. 
\end{rem}
\begin{lem}
\label{lem:highest-type-earns-positively}In any all-pay auction $\mathcal{G}(F_{1},F_{2})$,
if $\bar{v}_{i}=\sup\;\supp\;F_{i}>\inf\;\supp\;F_{i}=\underline{v}_{i}$,
then $x_{\sigma}<\bar{v}_{i}$ and $U_{i}(\bar{v}_{i}|\sigma)>0$
in any BNE $\sigma$. 
\end{lem}
\begin{proof}
The all-pay auction structure implies $x_{\sigma}\le\min\{\bar{v}_{1},\bar{v}_{2}\}$.
Suppose, to the contrary, that $x_{\sigma}=\bar{v}_{i}$. By Lemma
\ref{general-all-pay}, 
\[
H_{-i,\sigma}(x_{\sigma})=c_{-i,\sigma}+\int_{0}^{x_{\sigma}}\frac{1}{F_{i}^{-1}\left(H_{i,\sigma}(y)\right)}\,dy.
\]
Let $X_{-\bar{v}_{i}}$ be the set of bids by type $v_{i}<\bar{v}_{i}$
in $\sigma$. Thus, for $y\in X_{-\bar{v}_{i}}$, $F_{i}^{-1}\left(H_{i,\sigma}(y)\right)<\bar{v}_{i}$
or equivalently, $1/F_{i}^{-1}\left(H_{i,\sigma}(y)\right)>1/\bar{v}_{i}$.
This implies 
\begin{align*}
\int_{0}^{x_{\sigma}}\frac{1}{F_{i}^{-1}\left(H_{i,\sigma}(y)\right)}\,dy & >\int_{0}^{x_{\sigma}}\frac{1}{\bar{v}_{i}}\,dy=\int_{0}^{\bar{v}_{i}}\frac{1}{\bar{v}_{i}}\,dy=1
\end{align*}
Since $c_{-i,\sigma}\ge0$, it follows that $H_{-i,\sigma}(x_{\sigma})>1$,
which is impossible. Thus, type $\bar{v}_{i}$ earns a positive expected
payoff in any $\sigma$. 
\end{proof}

Finally, we establish a result useful for ranking $c_{1,\sigma}$
and $c_{2,\sigma}$ across different $\mathcal{G}(\tilde{F}_{1},\tilde{F}_{2})$.
\begin{lem}
\label{lem:key-identity} Let $\sigma$ be a BNE of any all-pay auction
$\mathcal{G}(\tilde{F}_{1},\tilde{F}_{2})$. Then 
\begin{equation}
\int_{c_{2,\sigma}}^{1}\frac{1}{\tilde{F}_{2}^{-1}(s)}ds=\int_{c_{1,\sigma}}^{1}\frac{1}{\tilde{F}_{1}^{-1}(s)}ds.\label{eq:key-identity}
\end{equation}
\end{lem}
\begin{proof}
By Lemma \ref{general-all-pay}, for $i=1,2$,
\[
H_{i,\sigma}'(\beta)=\frac{1}{\tilde{F}_{-i}^{-1}\bigl(H_{-i,\sigma}(\beta)\bigr)}.
\]
Trivially, 
\[
H_{1,\sigma}'(\beta)H_{2,\sigma}'(\beta)\equiv H_{2,\sigma}'(\beta)H_{1,\sigma}'(\beta)
\]
which can then be rewritten as 
\[
\frac{1}{\tilde{F}_{2}^{-1}\bigl(H_{2,\sigma}(\beta)\bigr)}H_{2,\sigma}'(\beta)=\frac{1}{\tilde{F}_{1}^{-1}\bigl(H_{1,\sigma}(\beta)\bigr)}H_{1,\sigma}'(\beta).
\]
Integrating both sides with respect to $\beta$ over $[0,x_{\sigma}]$,
we obtain the desired equation by a change of variables, because $H_{i,\sigma}(0)=c_{i,\sigma}$
and $H_{i,\sigma}(x_{\sigma})=1$. 
\end{proof}
This result can be interpreted as an alternative formulation of the
boundary condition that highlights the role of the lower limit of
integration. To see this, observe that when $\tilde{F}_{i}=\delta_{v_{i}^{*}}$,
equation \eqref{eq:key-identity} reduces to \eqref{eq:boundary-cond-one-sided}. 

The result above is useful for comparing two all-pay auctions involving
the same player. 
\begin{cor}
\label{cor:key-identity}Suppose player 1 participates in two all-pay
auctions, $\mathcal{G}_{a}$ and $\mathcal{G}_{b}$, with opponent
type distributions $\tilde{F}_{2,a}$ and $\tilde{F}_{2,b}$, respectively.
If $c_{1,\sigma^{a}}=c_{1,\sigma^{b}}$, then by Lemma \ref{lem:key-identity},
\[
\int_{c_{2,\sigma^{a}}}^{1}\frac{1}{\tilde{F}_{2,a}^{-1}(s)}ds=\int_{c_{2,\sigma^{b}}}^{1}\frac{1}{\tilde{F}_{2,b}^{-1}(s)}ds.
\]
\end{cor}
Finally, the following results, implied by Lemma \ref{lem:key-identity},
can be used to determine the positivity and the value of $c_{i,\sigma}$.
\begin{cor}
For any all-pay auction $\mathcal{G}(\tilde{F}_{1},\tilde{F}_{2})$
and any BNE $\sigma$ of $\mathcal{G}(\tilde{F}_{1},\tilde{F}_{2})$,
$c_{i,\sigma}=0$ if and only if 
\[
\int_{0}^{1}\frac{1}{\tilde{F}_{i}^{-1}(s)}ds\le\int_{0}^{1}\frac{1}{\tilde{F}_{-i}^{-1}(s)}ds.
\]
Equivalently, $c_{i,\sigma}>0$ if and only if 
\begin{equation}
\int_{0}^{1}\frac{1}{\tilde{F}_{i}^{-1}(s)}ds>\int_{0}^{1}\frac{1}{\tilde{F}_{-i}^{-1}(s)}ds.\label{eq:c_i_sigma_positive}
\end{equation}
Furthermore, if $c_{i,\sigma}>0$, then it is given by 
\begin{equation}
\int_{c_{i,\sigma}}^{1}\frac{1}{\tilde{F}_{i}^{-1}(s)}ds=\int_{0}^{1}\frac{1}{\tilde{F}_{-i}^{-1}(s)}ds.\label{eq:c_i_sigma_value}
\end{equation}
\end{cor}
\begin{proof}
We prove the second statement; the first follows by equivalence.

First, suppose \eqref{eq:c_i_sigma_positive} holds. We can conclude
$c_{i,\sigma}>0$ because, if $c_{i,\sigma}=0$, then by Lemma \ref{lem:key-identity}
\[
\int_{0}^{1}\frac{1}{\tilde{F}_{i}^{-1}(s)}ds=\int_{c_{-i,\sigma}}^{1}\frac{1}{\tilde{F}_{-i}^{-1}(s)}ds,
\]
which is impossible for any $c_{-i,\sigma}\ge0$. 

Second, it follows from above that if $c_{i,\sigma}>0$, then $c_{-i,\sigma}=0$.
By Lemma \ref{lem:key-identity}, \eqref{eq:c_i_sigma_value} holds,
which implies \eqref{eq:c_i_sigma_positive}. 
\end{proof}
In the corollary above, $\int_{0}^{1}\frac{1}{\tilde{F}_{i}^{-1}(s)}ds$
can be interpreted as the overall strength of player $-i$ relative
to player $i$. The corollary says that if \eqref{eq:c_i_sigma_positive}
is true, then player $-i$ has higher overall strength; consequently,
a positive measure of the weakest types of player $i$ bid zero and
can only earn zero expected payoff in the BNE. It also implies that
if player $-i$ becomes relatively stronger and/or player $i$ becomes
relatively weaker in the sense of FOSD, then the conclusion continues
to hold; conversely, if player $-i$ becomes weaker and/or player
$i$ becomes stronger in the sense of FOSD, then $c_{i,\sigma}$ decreases.

\section{The bribing model \protect}\label{sec:The-bribing-game}

It is immediate that, in any peaceful equilibrium, at most one bribe
can be accepted with probability one. Thus, any peaceful equilibrium
is pooling.

Consider a peaceful equilibrium in which all types of player 1 pool
on the bribe $\bar{b}\in[0,\underline{v}_{1}]$, which is accepted
by all types of player 2. In the equilibrium, type $v_{1}$'s payoff
is $v_{1}-\bar{b}$, whereas type $v_{2}$'s payoff is $\bar{b}$.
In such a peaceful equilibrium, there are two types of off-path deviations.
On one hand, rejection is off the path. On the other hand, there are
many unsent bribes.

We first consider player 2's rejection.

Suppose player 2 unexpectedly rejects the bribe $\bar{b}$. In the
continuation auction, player 1's type distribution is the same as
the prior. Because rejection is off the path, player 1 may hold arbitrary
beliefs about player 2's type distribution.

It is not profitable for all types of player 2 to reject $\bar{b}$
if and only if the highest type of player 2, namely $\bar{v}_{2}$,
earns an expected payoff no higher than $\bar{b}$, because for any
belief the expected payoff function of player 2 is increasing (by
Lemma \ref{general-all-pay}). Let the expected payoff of type $v_{2}$
in the continuation auction following player 2's rejection be $U_{2}(v_{2}|\sigma_{2})$
for a belief $\tilde{F}_{2}$. According to \citet{Zheng2019}, for
type $\bar{v}_{2}$, the highest possible $U_{2}(\bar{v}_{2}|\sigma_{2})$
is achieved when player 1 believes $v_{2}=\underline{v}_{2}$, and
the lowest possible $U_{2}(\bar{v}_{2}|\sigma_{2})$ is achieved when
player 1 believes $v_{2}=\bar{v}_{2}$. Let $\bar{\sigma}_{2}$ denote
a BNE induced by belief $v_{2}=\bar{v}_{2}$ and $\underline{\sigma}_{2}$
a BNE induced by belief $v_{2}=\underline{v}_{2}$. Thus peace implementability
requires that $U_{2}(\bar{v}_{2}|\bar{\sigma}_{2})\le\bar{b}$ whereas
peace security requires $U_{2}(\bar{v}_{2}|\underline{\sigma}_{2})\le\bar{b}$.
Equivalently, the former requires 
\begin{equation}
\bar{v}_{2}c_{1,\bar{\sigma}_{2}}\le\bar{b},\label{imple-cond-bidder2}
\end{equation}
whereas the latter requires 
\begin{equation}
\bar{v}_{2}-\underline{v}_{2}(1-c_{1,\underline{\sigma}_{2}})\le\bar{b},\label{secur-cond-bidder2}
\end{equation}
where the constants are defined by

\[
c_{1,\bar{\sigma}_{2}}:=\inf\{c_{1}\in[0,1]:\bar{v}_{2}\int_{c_{1}}^{1}\frac{1}{F_{1}^{-1}(s)}ds\le1\},
\]
\begin{gather*}
c_{1,\underline{\sigma}_{2}}:=\inf\{c_{1}\in[0,1]:\underline{v}_{2}\int_{c_{1}}^{1}\frac{1}{F_{1}^{-1}(s)}ds\le1\}.
\end{gather*}

Surprisingly, we find that, in contrast to \citet{Zheng2019}, peace
security is impossible in our model. 
\begin{thm}
\label{thm:security-impossible}Peace is not securable.\footnote{Although we assume atomless distributions for players' types, the
result holds for arbitrary independent type distributions, provided
player 1's distribution is non-degenerate, as implied by the proof.
(When player 1's type is degenerate, the highest-type rent identified
in Lemma \ref{lem:highest-type-earns-positively} vanishes, and, as
discussed below, peace can be secured.)} 
\end{thm}
\begin{proof}
See Appendix \ref{Proof-security-impossible}. 
\end{proof}
The intuition for this impossibility is as follows. As in \citet{Zheng2019},
peace security critically depends on the expected off-path payoffs
of the highest types of both players. However, the ability to offer
flexible off-path bribes in our model restricts candidate equilibrium
bribes to a unique value: the highest possible expected payoff that
type $\bar{v}_{2}$ earns in the continuation auction when she rejects
the bribe and is believed to be the lowest type. For player 1, paying
anything above is a waste of money. However, on top of this, player
1's highest type retains an incentive to extract some extra surplus
by offering an even lower off-path bribe. For peace to be secure,
such an undercut must lead to a positive probability of rejection
regardless of the belief; otherwise, it constitutes a profitable deviation.
Deterring such a deviation requires player 2 to be sufficiently strong
relative to player 1, specifically requiring $\bar{v}_{1}\le\underline{v}_{2}$.

The impossibility then arises from the nature of asymmetric information.
Consider a simplified case where player 1's type is degenerate at
$\bar{v}_{1}$ and player 2's type is degenerate at $\underline{v}_{2}$.
Suppose $2\bar{v}_{1}=\underline{v}_{2}>0$. In this setting, peace
can be secured if player 1 surrenders his entire valuation to player
2 (that is, no player has an incentive to deviate). As soon as player
1's type distribution becomes non-degenerate, the nature of all-pay
auctions ensures that the highest type always earns a positive expected
payoff (a form of information rent). This inherent surplus provides
a persistent incentive for player 1 to undercut the on-path bribe,
rendering absolute security impossible.

We next examine peace implementability by considering player 1's deviation
to off-path bribes.

Let $\pi_{1}(v_{1}|b,\tilde{F}_{1})$ be the expected payoff of type
$v_{1}$ from a deviation to an off-path bribe $b$ when player 2's
belief is $\tilde{F}_{1}$ which induces a continuation auction $\mathcal{G}(\tilde{F}_{1},\tilde{F}_{2})$. 
\begin{lem}
\label{slope-lower-than-1} For any off-path bribe $b$, for any $\tilde{F}_{1}$,
and any BNE $\sigma$ of any induced $\mathcal{G}(\tilde{F}_{1},\tilde{F}_{2})$,
$\pi_{1}(v_{1}|b,\tilde{F}_{1})$ is continuous and non-decreasing
in $v_{1}$, with a derivative no greater than one wherever it exists. 
\end{lem}
\begin{proof}
See Appendix \ref{proof-slope-lower-than-1}. 
\end{proof}
The preceding result implies that implementability critically depends
on the expected payoff of type $\underline{v}_{1}$ from deviations
to all possible off-path bribes, because the equilibrium payoff of
player 1 increases with type $v_{1}$ at a rate of one. That is,
if it is not profitable for type $\underline{v}_{1}$ to deviate to
some bribe $b$ under belief $\tilde{F}_{1}$, then no type $v_{1}>\underline{v}_{1}$
will profitably deviate to the same bribe under that belief.

Hence, examining peace implementability amounts to searching for the
belief $\tilde{F}_{1}$ that minimizes $\pi_{1}(\underline{v}_{1}|b,\tilde{F}_{1})$
for each $b$. Peace is implementable only if the maximum of these
minimized payoffs across all off-path bribes $b$ does not exceed
the equilibrium payoff of type $\underline{v}_{1}$. That is, peace
implementability requires 
\[
\max_{b}\min_{\tilde{F}_{1}}\pi_{1}(\underline{v}_{1}|b,\tilde{F}_{1})\le\underline{v}_{1}-\bar{b}.
\]

We continue to examine the expected payoff of type $v_{1}$ for different
$b$. 

We first consider the zero bribe. When the on-path bribe is positive,
the zero bribe is off-path and may admit multiple consistent replies.
This multiplicity arises because some types of player 2 may be indifferent
between accepting a zero bribe and rejecting it and bidding zero in
the continuation auction, as both actions yield zero payoff. 

We show below that among all beliefs and all consistent replies, full
rejection under the degenerate belief $\delta_{\underline{v}_{1}}$
minimizes the expected payoff of type $\underline{v}_{1}$ from deviating
to the zero bribe.
\begin{lem}
\label{lem:zero-bribe-min-reply} For any given belief $\tilde{F}_{1}$,
among all consistent replies to the off-path zero bribe, full rejection
minimizes the expected payoff of type $\underline{v}_{1}$.
\end{lem}
\begin{proof}
See Appendix \ref{Proof-of-zero-bribe}. 
\end{proof}
\begin{lem}
\label{lem:zero-bribe-worst-belief} Given that the consistent reply
to the off-path zero bribe is full rejection, among all beliefs, the
degenerate belief $\delta_{\underline{v}_{1}}$ minimizes the expected
payoff of type $\underline{v}_{1}$. 
\end{lem}
\begin{proof}
See Appendix \ref{Proof-of-worst-belief-to-zero-bribe}. 
\end{proof}

We next consider positive bribes. As mentioned above, for any off-path
bribe $b>0$ and any given $\tilde{F}_{1}$, if the rejection set
is non-empty, then by consistency it must be an interval: the lowest
type in the set is indifferent between accepting the positive bribe
and rejecting it, while any higher type strictly prefers rejection.
We use $a_{2,\sigma}(b)$ to denote this lowest rejecting type.\footnote{Although $a_{2,\sigma}(b)$ depends on $b$, we suppress the argument
to save notation in the proofs.} Since type $a_{2,\sigma}(b)$ must earn positive expected payoff
by bidding zero in the BNE $\sigma$ of the induced continuation auction,
we have $c_{1,\sigma}>0$ and $c_{2,\sigma}=0$, because $U_{2}(a_{2,\sigma}(b)|\sigma)=c_{1,\sigma}a_{2,\sigma}(b)$.
Furthermore, since type $\underline{v}_{1}$ is the lowest type of
player 1 and $c_{2,\sigma}=0$, he must earn zero expected payoff
in the continuation auction. We summarize these observations in the
following lemma.

\begin{lem}
\label{lem:c2sig-zero-p1-form}For any off-path bribe $b>0$, for
any given $\tilde{F}_{1}$, if the rejection set is non-empty, then
it is an interval $[a_{2,\sigma}(b),\bar{v}_{2}]$ for some $a_{2,\sigma}(b)$,
where $\sigma$ is the BNE of the continuation auction. In $\sigma$,
$c_{2,\sigma}=0$. Furthermore,
\begin{equation}
\pi_{1}(\underline{v}_{1}|b>0,\tilde{F}_{1})=F_{2}(a_{2,\sigma}(b))(\underline{v}_{1}-b).\label{eq:Pi-v1-low-bar}
\end{equation}
\end{lem}

The structure of $\pi_{1}(\underline{v}_{1}|b,\tilde{F}_{1})$ implies
that for each $b>0$, $\pi_{1}(\underline{v}_{1}|b,\tilde{F}_{1})$
is minimized by the lowest possible $a_{2,\sigma}(b)$ induced by
some $\tilde{F}_{1}$. Let $a_{2,\underline{\sigma}}(b)$ be the lowest
rejecting type of player 2 when player 2's belief is $\delta_{\underline{v}_{1}}$
and some BNE $\underline{\sigma}$ is induced. The following result
shows that for each $b>0$, $\pi_{1}(\underline{v}_{1}|b,\tilde{F}_{1})$
is minimized by the belief $\delta_{\underline{v}_{1}}$, or equivalently,
by $a_{2,\underline{\sigma}}(b)$. 

\begin{lem}
\label{lem:a2-lowest-highest} For any given off-path bribe $b>0$,
among all possible $\tilde{F}_{1}$ and any BNE $\sigma$ of any induced
$\mathcal{G}(\tilde{F}_{1},\tilde{F}_{2})$, $x_{\underline{\sigma}}\le x_{\sigma}$
and $a_{2,\underline{\sigma}}(b)\le a_{2,\sigma}(b)$. Hence, for
any $b>0$,
\[
\pi_{1}(\underline{v}_{1}|b,\delta_{\underline{v}_{1}})=F_{2}(a_{2,\underline{\sigma}}(b))(\underline{v}_{1}-b)=\min_{\tilde{F}_{1}}\;\pi_{1}(\underline{v}_{1}|b,\tilde{F}_{1})
\]
\end{lem}
\begin{proof}
See Appendix \ref{Proof-of-Lemma-a2-lowest-highest}. 
\end{proof}

Below we derive the conditions for $a_{2,\underline{\sigma}}(b)\le\bar{v}_{2}$
and show that the solution is unique.

Let $\Phi_{2}(v_{2}|x)$ be the probability distribution of $v_{2}$
conditional on $v_{2}\ge x$, i.e., 
\[
\Phi_{2}(v_{2}|x)=\frac{F_{2}(v_{2})-F_{2}(x)}{1-F_{2}(x)},
\]
and denote the inverse function by $\Phi_{2}^{-1}(\cdot|x)$. Let
\[
\mathcal{I}_{2}(x):=\int_{0}^{1}\frac{1}{\Phi_{2}^{-1}(s|x)}ds.
\]
As $x$ increases, $\Phi_{2}(v_{2}|x)$ becomes more FOSD. Hence,
$\mathcal{I}_{2}(x)$ is continuous and strictly decreasing in $x$.
Moreover, $\mathcal{I}_{2}(\underline{v}_{2})=\int_{0}^{1}\bigl(F_{2}^{-1}(s)\bigr)^{-1}ds.$

Suppose first that for some off-path $b>0$, the belief $\delta_{\underline{v}_{1}}$
induces a consistent reply with a non-empty rejection set $[a_{2,\underline{\sigma}},\bar{v}_{2}]$,
where $a_{2,\underline{\sigma}}\in(\underline{v}_{2},\bar{v}_{2})$.
In the continuation auction $\mathcal{G}(\delta_{\underline{v}_{1}},F_{2}(v_{2}|v_{2}\ge a_{2,\underline{\sigma}}))$,
we have $c_{2,\underline{\sigma}}=0$ by Lemma \ref{lem:c2sig-zero-p1-form}.
Applying the boundary condition $H_{1,\underline{\sigma}}(x_{\underline{\sigma}})=1$
to equation (\ref{eq:degenerate-boundary-cond}) from Lemma \ref{lem:one-sided-complete-info},
we obtain 
\begin{equation}
c_{1,\underline{\sigma}}=1-\underline{v}_{1}\mathcal{I}_{2}(a_{2,\underline{\sigma}}).\label{boundary-degen-new2}
\end{equation}

\begin{lem}
\label{lem:a2lowBar-exis-uniq}Consider an off-path bribe $b>0$.
Suppose the induced belief $\tilde{F}_{1}$ is $\delta_{\underline{v}_{1}}$.
If $b>\bar{v}_{2}-\underline{v}_{1}$, the consistent reply is full
acceptance; if $b\le\bar{v}_{2}-\underline{v}_{1}$, a consistent
reply with a non-empty rejection set $[a_{2,\underline{\sigma}}(b),\bar{v}_{2}]$
exists and is unique, and 
\begin{equation}
a_{2,\underline{\sigma}}(b)=\begin{cases}
\underline{v}_{2} & \mathrm{if}\;b\le b_{\min}\\
a_{2} & \mathrm{otherwise}
\end{cases}\label{eq:a2sigma-bribing}
\end{equation}
where\footnote{Observe that equation \eqref{boundary-degen-new2} implies $\underline{v}_{1}\mathcal{I}_{2}(\underline{v}_{2})\le1$,
so $b_{\min}\ge0$.}
\[
b_{\min}:=\underline{v}_{2}\left(1-\underline{v}_{1}\mathcal{I}_{2}(\underline{v}_{2})\right)
\]
and $a_{2}$ is the unique solution to
\begin{equation}
b=a_{2}\left(1-\underline{v}_{1}\mathcal{I}_{2}(a_{2})\right).\label{interpret-as-bidding-1st-price}
\end{equation}
Furthermore, $a_{2,\underline{\sigma}}(b)$ is continuous and strictly
increasing for $b>b_{\min}$. 
\end{lem}
\begin{proof}
See Appendix \ref{proof-a2-exis-uniq}. 
\end{proof}

By the result above, if $b_{\min}>0$, then the consistent reply is
full rejection for $b\in(0,b_{\min}]$. On the other hand, if $b_{\min}=0$,
then for all $b>0$, the consistent reply is partial rejection. Furthermore,
by Lemma \ref{lem:a2-lowest-highest}, the lowest expected payoff
of type $\underline{v}_{1}$ from the deviation is $\pi_{1}(\underline{v}_{1}|b>0,\delta_{\underline{v}_{1}})=F_{2}(a_{2,\underline{\sigma}}(b))(\underline{v}_{1}-b)$.
As $b$ approaches zero, the limit of the rejection set is full rejection
and $\lim_{b\downarrow0}\pi_{1}(\underline{v}_{1}|b>0,\delta_{\underline{v}_{1}})=0$. 

On the other hand, for the zero bribe, full rejection is always a
consistent reply since any type $v_{2}$ weakly prefers rejection
to acceptance. From Lemma \ref{lem:zero-bribe-worst-belief}, for
the zero off-path bribe, the expected payoff of type $\underline{v}_{1}$
from deviation is minimized at belief $\delta_{\underline{v}_{1}}$
and a full-rejection reply. Let $\pi_{1}(\underline{v}_{1}|b=0,\mathcal{G}(\delta_{\underline{v}_{1}},F_{2}))$
be the expected payoff. Observe that even though it is a minimum,
$\pi_{1}(\underline{v}_{1}|b=0,\mathcal{G}(\delta_{\underline{v}_{1}},F_{2}))$
may still be positive for some prior distributions. That is, $\pi_{1}(\underline{v}_{1}|b=0,\mathcal{G}(\delta_{\underline{v}_{1}},F_{2}))$
may not equal $\lim_{b\downarrow0}\pi_{1}(\underline{v}_{1}|b>0,\delta_{\underline{v}_{1}}).$
Hence, following the analysis above, a natural and useful question
(for examining the conditions for peace implementability) is whether
for some prior distributions, $\pi_{1}(\underline{v}_{1}|b=0,\mathcal{G}(\delta_{\underline{v}_{1}},F_{2}))$
could exceed the supremum of the minimum expected payoff from deviation
to positive bribes, namely $\sup_{b>0}\;\pi_{1}(\underline{v}_{1}|b,\delta_{\underline{v}_{1}})$. 

\begin{lem}
\label{lem:zero-bribe-payoff-is-not-max}Suppose $\underline{v}_{1}<\bar{v}_{2}$.
Suppose further that upon receiving an off-path zero bribe, player
2 forms a belief $\delta_{\underline{v}_{1}}$ and the consistent
reply is full rejection. Then\footnote{As explained in the proof, $\sup_{b>0}\;\pi_{1}(\underline{v}_{1}|b,\delta_{\underline{v}_{1}})$
can be rewritten as $\max_{b>0}\;\pi_{1}(\underline{v}_{1}|b,\delta_{\underline{v}_{1}})$.} 
\begin{equation}
\pi_{1}(\underline{v}_{1}|b=0,\mathcal{G}(\delta_{\underline{v}_{1}},F_{2}))\le\max_{b>0}\;\pi_{1}(\underline{v}_{1}|b,\delta_{\underline{v}_{1}}).\label{eq:zero-bribe-payoff-no-greater-positive-bribe}
\end{equation}
\end{lem}
\begin{proof}
See Appendix \ref{Proof-of-Lemma-zero-bribe-payoff-is-not-sup}.
\end{proof}

Intuitively, $\pi_{1}(\underline{v}_{1}|b=0,\mathcal{G}(\delta_{\underline{v}_{1}},F_{2}))$
is the expected payoff of type $\underline{v}_{1}$ from competing
directly with player 2 in the all-pay auction when his valuation is
of complete information. On the other hand, $\sup_{b>0}\;\pi_{1}(\underline{v}_{1}|b,\delta_{\underline{v}_{1}})$
is the highest possible expected payoff of type $\underline{v}_{1}$
when his valuation is also of complete information but he has the
opportunity to collude with player 2 through bribes. Thus, although
in the case of unsuccessful collusion he can only earn zero expected
payoff in the continuation auction, Lemma \ref{lem:zero-bribe-payoff-is-not-max}
confirms that this strategic option weakly dominates direct competition.

We are now ready to characterize the necessary and sufficient conditions
for peace implementability.
\begin{thm}
\label{thm:Peace-implementability}If $\bar{v}_{2}\le\underline{v}_{1}$,
then peace is implementable (through the zero bribe) if and only if
\[
c_{1,\bar{\sigma}_{2}}=0.
\]
If $\bar{v}_{2}>\underline{v}_{1}$, then peace is implementable if
and only if 
\begin{equation}
\bar{v}_{2}c_{1,\bar{\sigma}_{2}}+\tilde{\pi}_{1}^{*}(\underline{v}_{1})\le\underline{v}_{1},\label{imple-thm}
\end{equation}
where 
\begin{equation}
\tilde{\pi}_{1}^{*}(\underline{v}_{1}):=\max_{a_{2}\in[\underline{v}_{2},\bar{v}_{2}]}\;F_{2}(a_{2})\left(\underline{v}_{1}-a_{2}\left(1-\underline{v}_{1}\mathcal{I}_{2}(a_{2})\right)\right).\label{eq:pi1-tilde}
\end{equation}
Furthermore, all peaceful equilibria are robust.\footnote{As noted in \citet{Zheng2019}, the same equivalence between robustness
and implementability holds in his mediation model.} 
\end{thm}
\begin{proof}
See Appendix \ref{Proof-of-Theorem-implementability}. 
\end{proof}
\begin{rem}
It follows from the proof that when (\ref{imple-thm}) holds, any
$\bar{b}\in[\bar{v}_{2}c_{1,\bar{\sigma}_{2}},\,\underline{v}_{1}-\tilde{\pi}_{1}^{*}(\underline{v}_{1})]$
can sustain a peaceful equilibrium. Among all peaceful equilibria,
the maximum payoff of player 1 is achieved in the one with $\bar{b}=\bar{v}_{2}c_{1,\bar{\sigma}_{2}}$. 
\end{rem}
The theorem exhibits a natural dichotomy driven by the relative strength
of the two players. When player 1 is strong enough that his lowest
type exceeds player 2's highest type ($\bar{v}_{2}\leq\underline{v}_{1}$),
peace is implementable if and only if $c_{1,\bar{\sigma}_{2}}=0$.
This single condition ensures that player 2's highest type has no
incentive to reject the offer under the most optimistic belief of
player 1. No constraint on player 1's own deviation incentives is
needed: because even player 1's weakest type is stronger than any
type of player 2 (except for possibly type $\bar{v}_{2}$), all types
of player 1 can coerce player 2 into acceptance and are content to
pool.

When player 1 is weaker ($\bar{v}_{2}>\underline{v}_{1}$), he has
an incentive to gamble off the path, in particular, the weakest type.
Conceptually, the right-hand side of (\ref{imple-thm}), $\underline{v}_{1}$,
represents the social surplus from a peaceful settlement, available
for division between the two players. The inequality acts as an incentive
budget constraint: the total ``demand'' for surplus must not exceed
this available amount. The first term, $\bar{v}_{2}c_{1,\bar{\sigma}_{2}}$,
is the payment needed to deter player 2's highest type from rejecting---analogous
to the corresponding condition in the mediation model of \citet{Zheng2019}.
The second term, $\tilde{\pi}_{1}^{*}(\underline{v}_{1})$, captures
a new strategic force unique to the bribing  protocol: the lowest
type of player 1 can ``gamble'' by offering an off-path bribe. This
gamble is functionally equivalent to bidding in an FPA against player
2: if the bribe is accepted, type $\underline{v}_{1}$ wins at the
cost of the bribe; if rejected, he earns zero in the continuation
all-pay auction. More specifically, recall from Lemma \ref{lem:c2sig-zero-p1-form}
that $\pi_{1}(\underline{v}_{1}|b,\delta_{\underline{v}_{1}})=F_{2}(a_{2,\underline{\sigma}}(b))(\underline{v}_{1}-b)$.
Thus, maximizing $\pi_{1}(\underline{v}_{1}|b,\delta_{\underline{v}_{1}})$
is equivalent to calculating a best response in an FPA in which player
2 uses the bidding function in \eqref{interpret-as-bidding-1st-price}
(with $a_{2}$ replaced by $v_{2}$). The value $\tilde{\pi}_{1}^{*}(\underline{v}_{1})$
is precisely the maximum expected payoff from this off-path gamble.
If the combined demands of the two players exceed the available surplus
$\underline{v}_{1}$, no peaceful pooling equilibrium can be sustained:
the weakest type of player 1 has too much to gain from breaking the
pool.

\section{The requesting model \protect}\label{sec:requesting}

In this section we consider the opposite scenario in which player
1 requests a payment from player 2 in exchange for exiting the conflict.
Just as in the bribing model, peace is not implementable in any separating
equilibrium. Below we focus on robust peaceful equilibria, which can
only be pooling.

Suppose that a peaceful equilibrium exists. Denote the on-path request
by $\bar{r}$, which cannot be higher than $\underline{v}_{2}$.

Consider an off-path request $r$. As in the bribing model, for a
given belief, if it is optimal for some type $v_{2}$ to reject $r$,
then it is also optimal for any type $v_{2}'<v_{2}$ to reject the
request.\footnote{For any given belief, the expected payoff of player 2 is increasing
at a rate lower than one in the BNE of the associated continuation
auction.} Consequently, if the rejection set is non-empty, then it takes the
form of an interval $[\underline{v}_{2},\alpha_{2,\sigma}(r)]$ for
some $\alpha_{2,\sigma}(r)\le\bar{v}_{2}$. Clearly, it is optimal
for all types $v_{2}<r$ to reject the request $r$ and thus $\alpha_{2,\sigma}(r)\ge\min\{\bar{v}_{2},r\}$.

For any given belief $\tilde{F}_{1}$ and any BNE $\sigma$ of the
induced continuation auction, the expected payoff of type $\bar{v}_{1}$
from an off-path request $r$ is 
\begin{equation}
\pi_{1}(\bar{v}_{1}|r,\tilde{F}_{1})=F_{2}(\alpha_{2,\sigma}(r))(\bar{v}_{1}-x_{\sigma})+(1-F_{2}(\alpha_{2,\sigma}(r)))r.\label{eq:pi-v1bar-r}
\end{equation}
If the consistent reply in $\sigma$ is partial rejection (i.e., $\alpha_{2,\sigma}(r)<\bar{v}_{2}$),
then type $\alpha_{2,\sigma}(r)$ must be indifferent between rejecting
(and paying) $r$. In that case, because type $\alpha_{2,\sigma}(r)$
is the highest type of the rejection set, she must bid $x_{\sigma}$
in any BNE $\sigma$ of the continuation auction and win with probability
one and thus $x_{\sigma}=r$. If the consistent reply is full rejection
(i.e., $\alpha_{2,\sigma}(r)=\bar{v}_{2}$), then $x_{\sigma}\le r$. 
\begin{lem}
\label{lem-request-non-exist-bound} If $\bar{v}_{1}>2\underline{v}_{2}$,
then there is no peaceful equilibrium; if $\bar{v}_{1}\le2\underline{v}_{2}$,
then in any peaceful equilibrium, $\bar{r}\ge\bar{v}_{1}/2$. 
\end{lem}
\begin{proof}
See Appendix \ref{proof-request-non-exist-bound}. 
\end{proof}
Let $\Psi_{2}(v_{2}|x):=F_{2}(v_{2})/F_{2}(x)$ be the probability
distribution of $v_{2}$ conditional on $v_{2}\le x$ and let $\Psi_{2}^{-1}(\cdot|x)$
be the inverse function. Clearly $\int_{0}^{1}(\Psi_{2}^{-1}(s|\bar{v}_{2}))^{-1}ds=\int_{0}^{1}(F_{2}^{-1}(s))^{-1}ds$. 
\begin{lem}
\label{lem:request-no-equil-v1bar-too-low}No robust peaceful equilibria
exist if $\bar{v}_{1}\le\underline{v}_{2}$. 
\end{lem}
\begin{proof}
See Appendix \ref{proof-request-no-equil-v1bar-too-low}. 
\end{proof}
It follows from Lemmas \ref{lem-request-non-exist-bound} and \ref{lem:request-no-equil-v1bar-too-low}
that a robust peaceful equilibrium requires $\underline{v}_{2}<\bar{v}_{1}\le2\underline{v}_{2}$
and consequently $\underline{v}_{2}>0$. We assume these conditions
hold for the remainder of the analysis.

Consider $\mathcal{G}(\delta_{\bar{v}_{1}},F_{2})$ and denote a BNE
by $\bar{\sigma}^{*}$. Let $c_{2,\bar{\sigma}^{*}}$ be the unique
solution to 
\[
\bar{v}_{1}\int_{c_{2,\bar{\sigma}^{*}}}^{1}\left(F_{2}^{-1}\left(s\right)\right)^{-1}ds=1
\]
if $\bar{v}_{1}\int_{0}^{1}(F_{2}^{-1}(s))^{-1}ds>1$ and zero otherwise.
Let $x_{\bar{\sigma}^{*}}$ denote the highest bid in $\bar{\sigma}^{*}$.
By Lemma \ref{lem:one-sided-complete-info}, 
\[
x_{\bar{\sigma}^{*}}=\bar{v}_{1}(1-c_{2,\bar{\sigma}^{*}}).
\]

Let $\alpha_{2,\bar{\sigma}}(r)$ be the unique solution to
\begin{equation}
\bar{v}_{1}\int_{c_{2,\bar{\sigma}}}^{1}\left(\Psi_{2}^{-1}\left(s|\alpha_{2,\bar{\sigma}}(r)\right)\right)^{-1}ds=1\label{eq:request-a2sig-below-1}
\end{equation}
where $c_{2,\bar{\sigma}}$ is given by $r=\bar{v}_{1}(1-c_{2,\bar{\sigma}})$,
provided that a solution $\alpha_{2,\bar{\sigma}}(r)\in[\underline{v}_{2},\bar{v}_{2}]$
exists for \eqref{eq:request-a2sig-below-1}.\footnote{When $\alpha_{2,\bar{\sigma}}(r)$ decreases, $\Psi_{2}(\cdot|\alpha_{2,\bar{\sigma}}(r))$
becomes more FOSD. Stochastic dominance implies that the integral
$\int_{c_{2,\bar{\sigma}}}^{1}\left(\Psi_{2}^{-1}(s|\alpha_{2,\bar{\sigma}}(r))\right)^{-1}ds$
is increasing in $\alpha_{2,\bar{\sigma}}(r)$. Thus if for a given
$c_{2,\bar{\sigma}}$, (\ref{eq:request-a2sig-below-1}) admits an
$\alpha_{2,\bar{\sigma}}(r)\in[\underline{v}_{2},\bar{v}_{2}]$, then
it is unique.} As shown in the proof of Theorem~\ref{thm-requesting}, for $r<x_{\bar{\sigma}^{*}}$,
$\alpha_{2,\bar{\sigma}}(r)$ is well-defined and unique; at the boundary
$r=x_{\bar{\sigma}^{*}}$ equation (\ref{eq:request-a2sig-below-1})
yields $\alpha_{2,\bar{\sigma}}(x_{\bar{\sigma}^{*}})=\bar{v}_{2}$
(full rejection). Thus, the objective function in (\ref{eq:p1-v1bar-r-delta-v1bar})
is well-defined and continuous on the entire interval $[\underline{v}_{2},x_{\bar{\sigma}^{*}}]$.
Let 
\begin{equation}
r^{*}\in\arg\max_{r\in[\underline{v}_{2},x_{\bar{\sigma}^{*}}]}F_{2}(\alpha_{2,\bar{\sigma}}(r))(\bar{v}_{1}-r)+(1-F_{2}(\alpha_{2,\bar{\sigma}}(r)))r.\label{eq:p1-v1bar-r-delta-v1bar}
\end{equation}

We now characterize the necessary and sufficient conditions for the
existence of a robust peaceful equilibrium. 
\begin{thm}
\label{thm-requesting}In the requesting model, there exists a robust
peaceful equilibrium if and only if the following conditions are satisfied: 
\begin{enumerate}[label=(\roman{enumi}).]
\item $0<\underline{v}_{2}<\bar{v}_{1}\le2\underline{v}_{2};$ 
\item $c_{1,\underline{\sigma}_{2}}=0;$ 
\item $\underline{v}_{2}\ge F_{2}(\alpha_{2,\bar{\sigma}}(r^{*}))(\bar{v}_{1}-r^{*})+(1-F_{2}(\alpha_{2,\bar{\sigma}}(r^{*})))r^{*}.$ 
\end{enumerate}
In any robust peaceful equilibrium, $\bar{r}=\underline{v}_{2}$. 
\end{thm}
\begin{proof}
See Appendix \ref{sec:Proof-request-lmm-bidder-1}. 
\end{proof}
The intuition for $\bar{r}=\underline{v}_{2}$ is as follows. By the
D1 criterion, player 1's off-path belief is $\delta_{\bar{v}_{1}}$.
Since $\bar{v}_{1}>\underline{v}_{2}$, any request $r<\underline{v}_{2}$
induces full acceptance by player 2. If some type rejected, the lowest
rejecting type $\underline{v}_{2}$ would earn zero in the continuation
auction (because $\bar{v}_{1}>\underline{v}_{2}$), contradicting
the incentive to reject since acceptance yields $r>0$. Thus, any
$r<\underline{v}_{2}$ is fully accepted, yielding player 1 a payoff
of $r$. This makes deviation to $r\in(\bar{r},\underline{v}_{2})$
profitable if $\bar{r}<\underline{v}_{2}$. Since the equilibrium
request cannot exceed $\underline{v}_{2}$, it must be exactly $\underline{v}_{2}$.

We now turn to peace security. 

\begin{lem}
\label{lem:alpha2-sigma-ranking}For any off-path request $r$, for
any possible beliefs and any BNE $\sigma$ of the induced continuation
auction, 
\begin{equation}
\alpha_{2,\underline{\sigma}}\ge\alpha_{2,\sigma}\ge\alpha_{2,\bar{\sigma}}\quad\mathrm{and}\quad x_{\underline{\sigma}}\le x_{\sigma}\le x_{\bar{\sigma}}.\label{eq:request-security-a2-xsigma}
\end{equation}
\end{lem}
\begin{proof}
See Appendix \ref{Proof-alpha2-sigma-ranking}. 
\end{proof}
\begin{rem}
\label{rem:request-full-rej-acc}The inequalities in (\ref{eq:request-security-a2-xsigma})
imply that for any given off-path $r$,
\begin{enumerate}[label=(\roman{enumi}).]
\item If belief $\delta_{\bar{v}_{1}}$ induces full rejection as the consistent
reply, all beliefs do. 
\item If some belief induces full rejection as the consistent reply, so
does belief $\delta_{\underline{v}_{1}}$. 
\item If belief $\delta_{\underline{v}_{1}}$ induces full acceptance as
the consistent reply, all beliefs do. 
\item If some belief induces full acceptance as the consistent reply, so
does belief $\delta_{\bar{v}_{1}}$. 
\end{enumerate}
\end{rem}
We now characterize the necessary and sufficient conditions for peace
security. 
\begin{thm}
\label{thm:request-security}Suppose that there exists a robust peaceful
equilibrium in the requesting model. Then peace is securable if and
only if 
\begin{equation}
c_{1,\bar{\sigma}_{2}}=0,\text{\; and \;}\bar{v}_{1}-\underline{v}_{1}\le\underline{v}_{2}.\label{eq:request-security-player-2}
\end{equation}
\end{thm}
\begin{proof}
See Appendix \ref{subsec:Proof-of-request-security}. 
\end{proof}
Security requires that it is not profitable for the lowest type of
player 2 to reject the peaceful request when player 1 holds the the
most pessimistic belief $v_{2}=\bar{v}_{2}$. This requirement is
ensured by $c_{1,\bar{\sigma}_{2}}=0$. For condition $\bar{v}_{1}-\underline{v}_{1}\le\underline{v}_{2}$,
first consider the case $\underline{v}_{1}\leq\underline{v}_{2}$.
Player 1 can always deviate to a sufficiently high request that triggers
full rejection by all types of player 2 under any belief. Thus, if
player 2 holds the most optimistic belief $\delta_{\underline{v}_{1}}$,
the continuation auction is $\mathcal{G}(\delta_{\underline{v}_{1}},F_{2})$
and the highest bid equals $\underline{v}_{1}$. Type $\bar{v}_{1}$
can then bid $\underline{v}_{1}$ and earn $\bar{v}_{1}-\underline{v}_{1}$.
For peace to be securable, this deviation payoff must not exceed the
equilibrium payoff $\underline{v}_{2}$, yielding the condition $\bar{v}_{1}-\underline{v}_{1}\leq\underline{v}_{2}$.
By contrast, if $\underline{v}_{1}>\underline{v}_{2}$, then $\bar{v}_{1}-\underline{v}_{1}\le\underline{v}_{2}$
is satisfied automatically by the fact that $\bar{v}_{1}\leq2\underline{v}_{2}$.

\section{Conclusion \protect}\label{sec:Conclusion}

We analyze strategic bargaining for peaceful settlement before a conflict
escalates into an all-pay auction, comparing two distinct protocols:
offering a side payment and requesting one. By adapting the notions
of peace prospects from \citet{Zheng2019}, we explore how the ability
of a proposer to signal private information influences the implementability
and security of peace.

Our first main finding is that while peace is implementable through
a continuum of bribes, it cannot be secured, in contrast to the mediation
result in \citet{Zheng2019}. This impossibility stems from the proposer's
incentive to deviate to marginally lower bribes to exploit the receiver's
beliefs, a strategic dynamic that differs fundamentally from the mediation
framework of \citet{Zheng2019}. The proposer's tactic is no longer
to pretend weak and trick the opponent in the continuation auction
(as in the mediation model), but rather to pretend strong and force
the opponent to accept a lower bribe. 

Our second main finding is that in the requesting model, the possibility
of peace security can be restored under suitable conditions, but it
comes at the cost of sensitivity to minor payment perturbations. Unlike
the bribing model, where a continuum of bribes can sustain peace,
any robust peaceful equilibrium in the requesting model admits a single,
specific request, the lowest valuation of the paying player. 

Beyond these primary results, we show that the preservation of peace
implementability when players become stronger, originally identified
by \citet{Zheng2019}, holds in the bribing model but with a crucial
refinement. While FOSD suffices in the case of player 1, hazard-rate
dominance guarantees preservation when player 2 becomes stronger,
as it accounts for the effects of truncation on off-path gambling
payoffs. 

Finally, we also examine non-peaceful equilibria in the bribing model
and find that regular separating equilibria can be ruled out. This
result echoes \citet{rachmilevitch2013} but contrast with \citet{ES2004}.
We also show that in any regular non-peaceful equilibrium with finitely
many pooling regions, the number of on-path side payments is at most
two, and if it is two, then the higher one is accepted with certainty
and the lower one is rejected with a positive probability.

\newpage

\appendix

\section*{Appendix}

\section{Proofs for the bribing model}

\subsection{Proof of Theorem \ref{thm:security-impossible} \protect}\label{Proof-security-impossible}

For this proof, for convenience, let $b^{s}$ or $b'$ denote a candidate
equilibrium bribe.

\emph{Step 1: Uniqueness of the securing bribe.} Suppose two bribes
$b'>b^{s}$ can secure peace, with equilibria $P'$ and $\bar{P}$
respectively. By security, in $\bar{P}$, no type of player 2 rejects
$b^{s}$ under any belief; thus $U_{2}(v_{2}|\sigma_{2})\le b^{s}$
for all $v_{2}$ and any continuation equilibrium $\sigma_{2}$. Consider
an off-path deviation to $b\in(b^{s},b')$ in $P'$. Since $U_{2}(v_{2}|\sigma_{2})\le b^{s}<b$,
player 2 accepts $b$ with probability one against belief $F_{1}$.
This constitutes a profitable deviation for player 1, contradicting
the security of $b'$. Thus, only one bribe can secure peace.

\emph{Step 2: Identification of the candidate bribe.} The unique securing
bribe must be $b^{s}=U_{2}(\bar{v}_{2}|\underline{\sigma}_{2})$.
Suppose instead $b'>U_{2}(\bar{v}_{2}|\underline{\sigma}_{2})$. For
any off-path $b\in(U_{2}(\bar{v}_{2}|\underline{\sigma}_{2}),b')$,
player 2's maximum possible continuation payoff is bounded by $U_{2}(\bar{v}_{2}|\underline{\sigma}_{2})$
(as in \citet{Zheng2019}). Since $b$ exceeds this bound, player
2 accepts $b$ with probability one under any belief. This allows
player 1 to profitably deviate to $b$, contradicting security. Hence,
$b^{s}=U_{2}(\bar{v}_{2}|\underline{\sigma}_{2})$.

\emph{Step 3: Constraint on type supports.} We show that security
requires $\bar{v}_{1}\le\underline{v}_{2}$. Observe that $b^{s}=U_{2}(\bar{v}_{2}|\underline{\sigma}_{2})>0$,
so off-path bribes $b<b^{s}$ always exist. 
\begin{claim}
\label{claim-payoff-v2bar} If an off-path bribe $b\in(0,b^{s})$
is rejected under belief $\delta_{\bar{v}_{1}}$, then in the continuation
auction $\mathcal{G}(\delta_{\bar{v}_{1}},\tilde{F}_{2})$, any type
$v_{2}$ earns expected payoff at most $\bar{v}_{2}-\bar{v}_{1}$.
\end{claim}
\begin{proof}
If type $v_{2}$ rejects, the lowest type in the rejection set earns
positive payoff (since $b>0$), implying $c_{1,\bar{\sigma}}>0$ in
the BNE $\bar{\sigma}$. Consequently, type $\bar{v}_{1}$ in $\mathcal{G}(\delta_{\bar{v}_{1}},\tilde{F}_{2})$
earns zero payoff and bids up to $x_{\bar{\sigma}}=\bar{v}_{1}$.
Type $\bar{v}_{2}$ must bid $\bar{v}_{1}$ and earns exactly $\bar{v}_{2}-\bar{v}_{1}$.
Since payoffs are increasing in type, all $v_{2}$ earn at most $\bar{v}_{2}-\bar{v}_{1}$.
\end{proof}
Security requires that for any off-path $b<b^{s}$, there is a positive
probability of rejection under any belief. In particular, for belief
$\delta_{\bar{v}_{1}}$, type $\bar{v}_{2}$ must prefer rejection
to acceptance if the bribe is low enough. However, by Claim \ref{claim-payoff-v2bar},
rejecting yields at most $\bar{v}_{2}-\bar{v}_{1}$. If $\bar{v}_{1}>\underline{v}_{2}$,
then $\bar{v}_{2}-\bar{v}_{1}<\bar{v}_{2}-\underline{v}_{2}\le b^{s}$
(from Eq. \ref{secur-cond-bidder2}). Player 1 could then deviate
to $b\in(\max\{\bar{v}_{2}-\bar{v}_{1},0\},b^{s})$. Under belief
$\delta_{\bar{v}_{1}}$, type $\bar{v}_{2}$ would accept $b$ (since
$b>\bar{v}_{2}-\bar{v}_{1}$), violating the security requirement
that rejection must occur with positive probability. Thus, $\bar{v}_{1}\le\underline{v}_{2}$.

\emph{Step 4: Contradiction via highest-type rent.} Consider the auction
$\mathcal{G}(F_{1},\delta_{\underline{v}_{2}})$ with BNE $\underline{\sigma}_{2}$.
Since $\bar{v}_{1}\le\underline{v}_{2}$, the highest bid $x_{\underline{\sigma}_{2}}\le\bar{v}_{1}$.
Thus, $U_{2}(\bar{v}_{2}|\underline{\sigma}_{2})=\bar{v}_{2}-x_{\underline{\sigma}_{2}}\ge\bar{v}_{2}-\bar{v}_{1}$.
Since $b^{s}=U_{2}(\bar{v}_{2}|\underline{\sigma}_{2})$, we have
$b^{s}\ge\bar{v}_{2}-\bar{v}_{1}$.

Suppose strict inequality holds: $b^{s}>\bar{v}_{2}-\bar{v}_{1}$.
Player 1 could deviate to $b\in(\bar{v}_{2}-\bar{v}_{1},b^{s})$.
By Claim \ref{claim-payoff-v2bar}, under belief $\delta_{\bar{v}_{1}}$,
player 2 accepts $b$ with probability one. This is a profitable deviation,
contradicting security. Hence, $b^{s}=\bar{v}_{2}-\bar{v}_{1}$.

This implies $U_{2}(\bar{v}_{2}|\underline{\sigma}_{2})=\bar{v}_{2}-\bar{v}_{1}$,
which requires $x_{\underline{\sigma}_{2}}=\bar{v}_{1}$. In $\mathcal{G}(F_{1},\delta_{\underline{v}_{2}})$,
if the highest bid is $\bar{v}_{1}$, type $\bar{v}_{1}$ earns zero
expected payoff. This contradicts Lemma \ref{lem:highest-type-earns-positively},
which states that the highest type always earns positive rent in a
non-degenerate all-pay auction. Thus, peace is not securable.

\subsection{Proof of Lemma \ref{slope-lower-than-1} \protect}\label{proof-slope-lower-than-1}

For any given $b$, denote the rejection set by $\tilde{V}_{2}$.
From Lemma \ref{general-all-pay}, for any type $v_{1}\in[\underline{v}_{1},\bar{v}_{1}]$,
the expected payoff in any given BNE $\sigma$ of any continuation
auction $\mathcal{G}(\tilde{F}_{1},\tilde{F}_{2})$ is $U_{1}(v_{1}|b,\tilde{F}_{1})=\max_{\beta\in\mathbb{R}_{+}}\{H_{2,\sigma}(\beta)v_{1}-\beta\}.$
Since $U_{1}(v_{1}|b,\tilde{F}_{1})$ is the maximum of a family of
affine functions, $U_{1}(v_{1}|b,\tilde{F}_{1})$ is convex and thus
absolutely continuous and differentiable almost everywhere. The expected
payoff of any type $v_{1}\in[\underline{v}_{1},\bar{v}_{1}]$ from
the deviation is 
\[
\begin{aligned}\pi_{1}(v_{1}|b,\tilde{F}_{1})= & \left(1-P(v_{2}\in\tilde{V}_{2})\right)(v_{1}-b)+P\left(v_{2}\in\tilde{V}_{2}\right)\max_{\beta\in\mathbb{R}_{+}}\{H_{2,\sigma}(\beta)v_{1}-\beta\},\end{aligned}
\]
which is also absolutely continuous and differentiable almost everywhere.

Since $H_{2,\sigma}\le1$, whenever $\pi_{1}(v_{1}|b,\tilde{F}_{1})$
is differentiable, the slope is not higher than one. At those non-differentiable
points, the left and right derivatives of $\pi_{1}(v_{1}|b,\tilde{F}_{1})$
are not higher than one. Thus, $\pi_{1}(v_{1}|b,\tilde{F}_{1})$ increases
continuously at rates no greater than one.

\subsection{Proof of Lemma \ref{lem:zero-bribe-min-reply} }\label{Proof-of-zero-bribe}

For the zero bribe, let $A$ denote the set of accepting types and
$R$ the set of rejecting types. For type $\underline{v}_{1}$, the
expected payoff is $c_{2,\sigma}\underline{v}_{1}$ in a BNE $\sigma$
of the continuation auction by Remark \ref{rem:A-useful-interpretation}.
So his total expected payoff in the grand game from the zero off-path
bribe under any consistent reply $\varrho$ is: 
\[
\pi_{1}(\underline{v}_{1}\mid b=0,\varrho)=\left(F_{2}(A)+F_{2}(R)c_{2,\sigma}\right)\underline{v}_{1}.
\]
 Define the \emph{total inactive probability} $P_{0}(\varrho):=F_{2}(A)+F_{2}(R)\,c_{2,\sigma}$.
We show that full rejection minimizes $P_{0}(\varrho)$ among all
consistent replies $\varrho$. 

For the zero bribe, full rejection is always a consistent reply. Let
$\sigma^{\mathrm{full}}$ be a BNE of the continuation auction under
full rejection $\varrho_{\mathrm{full}}$. If $c_{1,\sigma^{\mathrm{full}}}>0$,
then $c_{2,\sigma^{\mathrm{full}}}=0$ and type $\underline{v}_{1}$'s
expected payoff in $\sigma^{\mathrm{full}}$ is zero, with which the
current lemma is automatically true. Hence, below we focus on the
case with $c_{1,\sigma^{\mathrm{full}}}=0$.

In addition to full rejection, there could be consistent replies with
partial rejection. In any BNE $\sigma^{\mathrm{part}}$ of the continuation
auction induced by partial rejection, $c_{1,\sigma^{\mathrm{part}}}=0$.\footnote{Partial rejection means an acceptance set $A$ with $F_{2}(A)>0$.
By consistency, any type in $A$ must weakly prefer accepting $b=0$
to entering the continuation auction. If $c_{1,\sigma}>0$, then player
1 bids zero with positive probability, so every type $v_{2}>\underline{v}_{2}$
earns at least an expected payoff $c_{1,\sigma}v_{2}>0$---strictly
better than accepting $b=0$.}

To save notation below, we let $c_{2}^{\mathrm{part}}:=c_{2,\sigma^{\mathrm{part}}}$
and $c_{2}^{\mathrm{full}}:=c_{2,\sigma^{\mathrm{full}}}$.

Now consider a consistent partial-rejection reply $\varrho_{\mathrm{part}}$.
Let $p:=F_{2}(A)\in(0,1)$. Let $F_{2,\mathrm{part}}$ be the distribution
function of $v_{2}$ conditional on the rejection set $R$. Let $F_{2,\mathrm{part}}^{-1}$
denote the inverse function of $F_{2,\mathrm{part}}$. 

From above, $c_{1,\sigma^{\mathrm{part}}}=0$. In the nontrivial case
of full rejection, $c_{1,\sigma^{\mathrm{full}}}=0$. Since we are
comparing two continuation auctions with the same $\tilde{F}_{1}$,
Corollary \ref{cor:key-identity} gives 
\begin{equation}
\int_{c_{2}^{\mathrm{part}}}^{1}\frac{1}{F_{2,\mathrm{part}}^{-1}(s)}ds=\int_{c_{2}^{\mathrm{full}}}^{1}\frac{1}{F_{2}^{-1}(t)}dt.\label{eq:zero-bribe-min-key}
\end{equation}
For a rejection set with $F_{2}$-measure $1-p$, consider the corresponding
upper tail of $F_{2}$ with same measure $1-p$.\footnote{We do not claim there must exist a consistent reply with the upper
tail set being the set of rejecting types. We only use its statistical
property.} Let $F_{\mathrm{ut}}(v_{2}):=\left(F_{2}(v_{2})-p\right)/\left(1-p\right)$
be the distribution of the upper tail. Then the inverse $F_{\mathrm{ut}}^{-1}(s)=F_{2}^{-1}\bigl(p+(1-p)s\bigr)$.
Clearly, $F_{\mathrm{ut}}$ first-order stochastically dominates the
conditional distribution of any rejection set with $F_{2}$-measure
$1-p$. Thus, for all $s\in[0,1]$, 
\[
F_{2,\mathrm{part}}^{-1}(s)\le F_{2}^{-1}\bigl(p+(1-p)s\bigr),
\]
or equivalently
\[
\frac{1}{F_{2,\mathrm{part}}^{-1}(s)}\ge\frac{1}{F_{2}^{-1}\bigl(p+(1-p)s\bigr)}.
\]
Integrating both sides over the same interval $[c_{2}^{\mathrm{part}},1]$,
we have 
\begin{equation}
\int_{c_{2}^{\mathrm{part}}}^{1}\frac{1}{F_{2,\mathrm{part}}^{-1}(s)}ds\ge\int_{c_{2}^{\mathrm{part}}}^{1}\frac{1}{F_{2}^{-1}\bigl(p+(1-p)s\bigr)}ds=\frac{1}{1-p}\int_{p+(1-p)c_{2}^{\mathrm{part}}}^{1}\frac{1}{F_{2}^{-1}(t)}dt.\label{eq:zero-bribe-inter}
\end{equation}
By equation \eqref{eq:zero-bribe-min-key}, the LHS of \eqref{eq:zero-bribe-inter}
equals $\int_{c_{2}^{\mathrm{full}}}^{1}\frac{1}{F_{2}^{-1}(t)}dt$.
Since $P_{0}(\varrho_{\mathrm{part}})=p+(1-p)c_{2}^{\mathrm{part}}$,
the LHS and the RHS of \eqref{eq:zero-bribe-inter} can be rewritten
as
\[
(1-p)\int_{c_{2}^{\mathrm{full}}}^{1}\frac{1}{F_{2}^{-1}(t)}dt\ge\int_{P_{0}(\varrho_{\mathrm{part}})}^{1}\frac{1}{F_{2}^{-1}(t)}dt.
\]
Since $1-p<1$, it follows that 
\[
\int_{c_{2}^{\mathrm{full}}}^{1}\frac{1}{F_{2}^{-1}(t)}dt>(1-p)\int_{c_{2}^{\mathrm{full}}}^{1}\frac{1}{F_{2}^{-1}(t)}dt\ge\int_{P_{0}(\varrho_{\mathrm{part}})}^{1}\frac{1}{F_{2}^{-1}(t)}dt.
\]
Since the integrand is positive on the support, the integral is strictly
decreasing in its lower limit, and we have 
\[
P_{0}(\varrho_{\mathrm{part}})>c_{2}^{\mathrm{full}}=P_{0}(\varrho_{\mathrm{full}}).
\]
Therefore every consistent reply with partial rejection gives type
$\underline{v}_{1}$ a strictly higher payoff than full rejection.

\subsection{Proof of Lemma \ref{lem:zero-bribe-worst-belief} }\label{Proof-of-worst-belief-to-zero-bribe}

Let $\tilde{F}_{1}$ be any belief over player 1's type with support
contained in $[\underline{v}_{1},\bar{v}_{1}]$. Let $\sigma(\tilde{F}_{1})$
denote a BNE of the full-rejection continuation auction $\mathcal{G}(\tilde{F}_{1},F_{2})$,
and let $c_{2}(\tilde{F}_{1}):=c_{2,\sigma(\tilde{F}_{1})}$ be player
2's bid mass at zero. 

We show that for every such belief, 
\[
c_{2}(\tilde{F}_{1})\ge c_{2}(\delta_{\underline{v}_{1}}).
\]

Fix a belief $\tilde{F}_{1}$ with support in $[\underline{v}_{1},\bar{v}_{1}]$
and let $\sigma(\tilde{F}_{1})$ be a BNE of $\mathcal{G}(\tilde{F}_{1},F_{2})$.
Write $c_{1}:=c_{1,\sigma(\tilde{F}_{1})}$ and $c_{2}:=c_{2,\sigma(\tilde{F}_{1})}$.
By Lemma \ref{lem:key-identity}, 
\[
\int_{c_{2}}^{1}\frac{1}{F_{2}^{-1}(s)}ds=\int_{c_{1}}^{1}\frac{1}{\tilde{F}_{1}^{-1}(s)}ds,
\]
and by Lemma \ref{general-all-pay}, $c_{1}c_{2}=0$. We distinguish
two cases. 

\noindent\textbf{Case 1: $c_{2}>0$.} Then $c_{1}=0$ and the equation
above reads 
\[
\int_{c_{2}}^{1}\frac{1}{F_{2}^{-1}(s)}ds=\int_{0}^{1}\frac{1}{\tilde{F}_{1}^{-1}(s)}ds.
\]
Because $\tilde{F}_{1}^{-1}(s)\ge\underline{v}_{1}$ for all $s\in(0,1)$,
\[
\int_{0}^{1}\frac{1}{\tilde{F}_{1}^{-1}(s)}ds\le\int_{0}^{1}\frac{1}{\underline{v}_{1}}ds=\frac{1}{\underline{v}_{1}},
\]
with equality if and only if $\tilde{F}_{1}=\delta_{\underline{v}_{1}}$.
Hence 
\[
\int_{c_{2}}^{1}\frac{1}{F_{2}^{-1}(s)}ds\le\frac{1}{\underline{v}_{1}}.
\]
For the degenerate belief $\delta_{\underline{v}_{1}}$, let $c_{1}^{\star}$
and $c_{2}^{\star}$ denote the corresponding masses at zero. If $c_{2}^{\star}>0$,
then $c_{1}^{\star}=0$ and 
\[
\int_{c_{2}^{\star}}^{1}\frac{1}{F_{2}^{-1}(s)}ds=\int_{0}^{1}\frac{1}{\underline{v}_{1}}ds=\frac{1}{\underline{v}_{1}}.
\]
Combining the two equations above, we have $c_{2}\ge c_{2}^{\star}$.

\noindent\textbf{Case 2: $c_{2}=0$.} Then type $\underline{v}_{1}$'s
expected payoff under full rejection is $U(\underline{v}_{1}|\sigma)=c_{2,\sigma}\underline{v}_{1}=0$
as in Remark \ref{rem:boundary-cond-one-sided} (i), which trivially
satisfies the lemma.

In all cases, $c_{2}(\tilde{F}_{1})\ge c_{2}(\delta_{\underline{v}_{1}})$,
as claimed. This completes the proof.

\subsection{Proof of Lemma \ref{lem:a2-lowest-highest} \protect}\label{Proof-of-Lemma-a2-lowest-highest}

We first consider non-empty rejection sets.

From Lemma \ref{general-all-pay}, with $c_{2,\sigma}=0$, the boundary
condition $H_{2,\sigma}(x_{\sigma})=1$ becomes
\[
1=\int_{0}^{x_{\sigma}}\frac{1}{\tilde{F}_{1}^{-1}(H_{1,\sigma}(y))}\,dy.
\]
With belief $\delta_{\underline{v}_{1}}$, $x_{\underline{\sigma}}=\underline{v}_{1}.$
For any other belief $\tilde{F}_{1}$, since $\tilde{F}_{1}^{-1}(s)\geq\underline{v}_{1}$
for all $s$, we have 
\[
1=\int_{0}^{x_{\sigma}}\frac{1}{\tilde{F}_{1}^{-1}(H_{1,\sigma}(y))}dy\leq\int_{0}^{x_{\sigma}}\frac{1}{\underline{v}_{1}}dy=\frac{x_{\sigma}}{\underline{v}_{1}}.
\]
Therefore $x_{\sigma}\geq\underline{v}_{1}=x_{\underline{\sigma}}$. 

Suppose, for contradiction, $a_{2,\underline{\sigma}}>a_{2,\sigma}$.

Given an off-path bribe $b$, consider $\underline{\sigma}$ and some
$\sigma$. The indifference condition for the lowest rejecting type
implies $a_{2,\underline{\sigma}}c_{1,\underline{\sigma}}=b$ and
$a_{2,\sigma}c_{1,\sigma}\ge b$ (with equality if interior). Hence,
$a_{2,\underline{\sigma}}>a_{2,\sigma}$ implies 
\begin{equation}
c_{1,\underline{\sigma}}<c_{1,\sigma}.\label{rej-indifference}
\end{equation}
We now derive the reverse inequality using the key identity in Lemma
\ref{lem:key-identity}.

Since both rejection sets are non-empty, Lemma \ref{lem:c2sig-zero-p1-form}
gives $c_{2,\underline{\sigma}}=c_{2,\sigma}=0$. Let $\tilde{F}_{2}^{\underline{\sigma}}$
and $\tilde{F}_{2}^{\sigma}$ denote the distributions of $v_{2}$
conditional on the respective rejection sets $[a_{2,\underline{\sigma}},\bar{v}_{2}]$
and $[a_{2,\sigma},\bar{v}_{2}]$, respectively. Since $a_{2,\underline{\sigma}}>a_{2,\sigma}$,
$\tilde{F}_{2}^{\underline{\sigma}}$ first-order stochastically dominates
$\tilde{F}_{2}^{\sigma}$. By Lemma \ref{lem:key-identity} (with
$c_{2,\cdot}=0$): 
\[
\frac{1-c_{1,\underline{\sigma}}}{\underline{v}_{1}}=\int_{c_{1,\underline{\sigma}}}^{1}\frac{1}{\underline{v}_{1}}ds=\int_{0}^{1}\frac{1}{(\tilde{F}_{2}^{\underline{\sigma}})^{-1}(s)}ds\le\int_{0}^{1}\frac{1}{(\tilde{F}_{2}^{\sigma})^{-1}(s)}ds=\int_{c_{1,\sigma}}^{1}\frac{1}{\tilde{F}_{1}^{-1}(s)}ds.
\]
Since $\tilde{F}_{1}^{-1}(s)\geq\underline{v}_{1}$ for all $s$,
for the RHS of the above inequality we further have 
\[
\int_{c_{1,\sigma}}^{1}\frac{1}{\tilde{F}_{1}^{-1}(s)}ds\leq\int_{c_{1,\sigma}}^{1}\frac{1}{\underline{v}_{1}}ds=\frac{1-c_{1,\sigma}}{\underline{v}_{1}}.
\]
Combining these inequalities yields $1-c_{1,\underline{\sigma}}\le1-c_{1,\sigma}$,
i.e., $c_{1,\underline{\sigma}}\ge c_{1,\sigma}$, contradicting (\ref{rej-indifference})
and thus the supposition. Thus $a_{2,\underline{\sigma}}\le a_{2,\sigma}$
when the rejection sets are not empty. 

Finally, if the rejection set is empty under $\delta_{\underline{v}_{1}}$
(i.e., $a_{2,\underline{\sigma}}=\bar{v}_{2}$), it must be empty
under any other belief. Observe that if some $\sigma$ (induced by
some other belief than $\delta_{\underline{v}_{1}}$) had a non-empty
rejection set, type $\bar{v}_{2}$ would earn $U_{2}(\bar{v}_{2}|\sigma)>b$.
Since $x_{\underline{\sigma}}\le x_{\sigma}$, type $\bar{v}_{2}$
could bid $x_{\underline{\sigma}}$ in $\underline{\sigma}$ to earn
strictly more than $b$, implying rejection is strictly profitable
under $\delta_{\underline{v}_{1}}$ as well and thus the rejection
set cannot be empty for belief $\delta_{\underline{v}_{1}}$. That
is, $a_{2,\underline{\sigma}}\le a_{2,\sigma}$ even in this case.

\subsection{Proof of Lemma \ref{lem:a2lowBar-exis-uniq} \protect}\label{proof-a2-exis-uniq}

For type $a_{2,\underline{\sigma}}$ to be indifferent between accepting
and rejecting the bribe, $U_{2}(a_{2,\underline{\sigma}}|b,\delta_{\underline{v}_{1}})=c_{1,\underline{\sigma}}a_{2,\underline{\sigma}}=b$.
Combining this with (\ref{boundary-degen-new2}), the indifference
condition can be written as 
\begin{equation}
b=a_{2,\underline{\sigma}}\left(1-\underline{v}_{1}\mathcal{I}_{2}(a_{2,\underline{\sigma}})\right).\label{a2-indifference-new}
\end{equation}

Let $g(a)=a(1-\underline{v}_{1}\mathcal{I}_{2}(a))$. We examine the
monotonicity of $g(a)$ for $a\in[\underline{v}_{2},\bar{v}_{2}]$.
As $a$ increases, the term $(1-\underline{v}_{1}\mathcal{I}_{2}(a))$
is strictly increasing because $\mathcal{I}_{2}(a)$ is strictly decreasing
in $a$. Since $a$ is also strictly increasing in $a$ and $(1-\underline{v}_{1}\mathcal{I}_{2}(a))\ge c_{1,\underline{\sigma}}$
is nonnegative for $a\in[\underline{v}_{2},\bar{v}_{2}]$, $g(a)$
is strictly increasing in $a$.

The range of $g(a)$ on $[\underline{v}_{2},\bar{v}_{2}]$ is $[g(\underline{v}_{2}),g(\bar{v}_{2})]$,
where 
\[
g(\underline{v}_{2})=\underline{v}_{2}(1-\underline{v}_{1}\mathcal{I}_{2}(\underline{v}_{2}))\quad\text{and}\quad g(\bar{v}_{2})=\bar{v}_{2}(1-\underline{v}_{1}/\bar{v}_{2})=\bar{v}_{2}-\underline{v}_{1}.
\]

Since $g(a)$ is strictly increasing, for any $b\in[g(\underline{v}_{2}),\bar{v}_{2}-\underline{v}_{1}]$,
there exists a unique $a_{2,\underline{\sigma}}\in(\underline{v}_{2},\bar{v}_{2})$
such that $g(a_{2,\underline{\sigma}})=b$. This also shows that $a_{2,\underline{\sigma}}$
is strictly increasing in $b$ in this range. It is also clear that
$a_{2,\underline{\sigma}}$ is continuous in $b$. 

Next we consider the cases of full rejection and empty rejection set.

Observe that $U_{2}(a_{2,\underline{\sigma}}|b,\delta_{\underline{v}_{1}})=c_{1,\underline{\sigma}}a_{2,\underline{\sigma}}$.
So if $b<g(\underline{v}_{2})$, then the consistent reply is full
rejection, i.e., $a_{2,\underline{\sigma}}=\underline{v}_{2}$.

Finally, from above, a non-empty rejection set requires $b\le\bar{v}_{2}-\underline{v}_{1}$.
Thus, if $b>\bar{v}_{2}-\underline{v}_{1}$, the rejection set is
empty.

Since $\mathcal{I}_{2}(\underline{v}_{2})=\int_{0}^{1}(F_{2}^{-1}(s))^{-1}ds$
and $b_{\min}:=\underline{v}_{2}\left(1-\underline{v}_{1}\mathcal{I}_{2}(\underline{v}_{2})\right)=g(\underline{v}_{2})$,
this completes the proof.

\subsection{Proof of Lemma \ref{lem:zero-bribe-payoff-is-not-max} }\label{Proof-of-Lemma-zero-bribe-payoff-is-not-sup}

From Lemma \ref{lem:a2lowBar-exis-uniq}, if $b>\bar{v}_{2}-\underline{v}_{1}$,
the consistent reply is full acceptance; if $b=\bar{v}_{2}-\underline{v}_{1}$,
only type $\bar{v}_{2}$ would possibly reject $b$. So for $b\ge\bar{v}_{2}-\underline{v}_{1}$,
the maximum payoff of type $\underline{v}_{1}$ is $\underline{v}_{1}-(\bar{v}_{2}-\underline{v}_{1})$.
 If $b\le\bar{v}_{2}-\underline{v}_{1}$, then a consistent reply
with a non-empty rejection set $[a_{2,\underline{\sigma}}(b),\bar{v}_{2}]$
exists and is unique. Furthermore, $a_{2,\underline{\sigma}}(b)$
is strictly increasing for $b>b_{\min}$ and equals $\underline{v}_{2}$
for $b\le b_{\min}$. So the expected payoff of type $\underline{v}_{1}$,
$\pi_{1}(\underline{v}_{1}|b,\delta_{\underline{v}_{1}})=F_{2}(a_{2,\underline{\sigma}}(b))(\underline{v}_{1}-b)$
(by Lemma \ref{lem:a2-lowest-highest}), is continuous  everywhere.
When $b$ approaches zero, the limit of the rejection set is full
rejection and $\lim_{b\downarrow0}\pi_{1}(\underline{v}_{1}|b>0,\delta_{\underline{v}_{1}})=0$.
Therefore, we can conclude that for $b>0$, there is a maximum. 

Let $\sigma^{*}$ be the unique BNE of the continuation auction $\mathcal{G}(\delta_{\underline{v}_{1}},F_{2})$
following full rejection of the off-path zero bribe. Since $c_{1,\sigma^{*}}c_{2,\sigma^{*}}=0$,
it follows from Lemma \ref{lem:one-sided-complete-info} that 
\[
\pi_{1}(\underline{v}_{1}|b=0,\mathcal{G}(\delta_{\underline{v}_{1}},F_{2}))=U_{1}(\underline{v}_{1}\mid\sigma^{*})=c_{2,\sigma^{*}}\cdot\underline{v}_{1}.
\]
From Remark \ref{rem:boundary-cond-one-sided}, this implies 
\[
c_{1,\sigma^{*}}=1-\underline{v}_{1}\int_{c_{2,\sigma^{*}}}^{1}\frac{1}{F_{2}^{-1}(s)}ds.
\]

If $c_{1,\sigma^{*}}>0$, then $c_{2,\sigma^{*}}=0$, which implies
$U_{1}(\underline{v}_{1}\mid\sigma^{*})=0$. In this case, the inequality
\eqref{eq:zero-bribe-payoff-no-greater-positive-bribe} holds trivially.

Now, suppose $c_{1,\sigma^{*}}=0$. The boundary condition is
\begin{equation}
\int_{c_{2,\sigma^{*}}}^{1}\frac{1}{F_{2}^{-1}(s)}ds=\frac{1}{\underline{v}_{1}}.\label{eq:zero-bribe-c1}
\end{equation}

Observe that $\mathcal{I}_{2}(x)=\int_{0}^{1}1/\Phi_{2}^{-1}(s|x)ds$
is continuous and strictly decreasing on $[\underline{v}_{2},\bar{v}_{2}]$,
with 
\begin{equation}
\mathcal{I}_{2}(\bar{v}_{2})=\frac{1}{\bar{v}_{2}}<\frac{1}{\underline{v}_{1}}.\label{eq:I-2-min}
\end{equation}

Let $\underline{\mathcal{I}}_{2}=\lim_{x\downarrow\underline{v}_{2}}\mathcal{I}_{2}(x)$. 

If $\underline{\mathcal{I}}_{2}>1/\underline{v}_{1}$, then by continuity
and the fact that $\mathcal{I}_{2}(\bar{v}_{2})<1/\underline{v}_{1}$,
there exists a unique $a^{*}\in(\underline{v}_{2},\bar{v}_{2})$ such
that $\mathcal{I}_{2}(a^{*})=1/\underline{v}_{1}$. Observe that $\mathcal{I}_{2}(x)$
can be rewritten as 
\[
\mathcal{I}_{2}(x)=\frac{1}{1-F_{2}(x)}\int_{F_{2}(x)}^{1}\frac{1}{F_{2}^{-1}(s)}ds.
\]
So the condition $\mathcal{I}_{2}(a^{*})=1/\underline{v}_{1}$ is
equivalent to 
\[
\int_{F_{2}(a^{*})}^{1}\frac{1}{F_{2}^{-1}(s)}ds=\frac{1-F_{2}(a^{*})}{\underline{v}_{1}}.
\]
Comparing this with \eqref{eq:zero-bribe-c1}, we conclude $F_{2}(a^{*})>c_{2,\sigma^{*}}$.

For any $a_{2}\in(a^{*},\bar{v}_{2}]$, the bribe $b(a_{2})=a_{2}(1-\underline{v}_{1}\mathcal{I}_{2}(a_{2}))$
is positive and induces a valid consistent reply. By Lemma \ref{lem:c2sig-zero-p1-form},
$\pi_{1}(\underline{v}_{1}\mid b(a_{2}))=F_{2}(a_{2})\bigl(\underline{v}_{1}-b(a_{2})\bigr)$.
As $a_{2}\downarrow a^{*}$, we have $b(a_{2})\to0$ and $F_{2}(a_{2})\to F_{2}(a^{*})$,
which implies 
\[
\lim_{a_{2}\downarrow a^{*}}\pi_{1}(\underline{v}_{1}\mid b(a_{2}))=F_{2}(a^{*})\underline{v}_{1}>c_{2,\sigma^{*}}\underline{v}_{1}.
\]
This establishes that $\max_{b>0}\pi_{1}(\underline{v}_{1}\mid b,\delta_{\underline{v}_{1}})>U_{1}(\underline{v}_{1}\mid\sigma^{*})$.

Finally, suppose $\underline{\mathcal{I}}_{2}\leq1/\underline{v}_{1}$.
Then $\mathcal{I}_{2}(x)\leq1/\underline{v}_{1}$ for all $x$. Comparing
the definition of $\mathcal{I}_{2}(x)$ with \eqref{eq:zero-bribe-c1},
we have
\[
\lim_{x\downarrow\underline{v}_{2}}\mathcal{I}_{2}(x)=\int_{0}^{1}\frac{1}{F_{2}^{-1}(s)}ds\le\int_{c_{2,\sigma^{*}}}^{1}\frac{1}{F_{2}^{-1}(s)}ds=\frac{1}{\underline{v}_{1}}.
\]
This implies $c_{2,\sigma^{*}}=0$. So again the inequality \eqref{eq:zero-bribe-payoff-no-greater-positive-bribe}
holds trivially.

In all cases, the inequality holds and the proof is complete.

\subsection{Proof of Theorem \ref{thm:Peace-implementability} \protect}\label{Proof-of-Theorem-implementability}

Peace is implementable if and only if there exists a bribe $\bar{b}$
such that no type of player 1 deviates and player 2 accepts under
some belief system. By Lemma \ref{slope-lower-than-1}, player 1's
deviation payoff $\pi_{1}(v_{1}|b,\tilde{F}_{1})$ is non-decreasing
in $v_{1}$ with slope at most 1. Since the equilibrium payoff $v_{1}-\bar{b}$
has slope 1, it suffices to check type $\underline{v}_{1}$. Thus,
player 1 does not deviate if and only if: 
\[
\max_{b}\min_{\tilde{F}_{1}}\pi_{1}(\underline{v}_{1}|b,\tilde{F}_{1})\le\underline{v}_{1}-\bar{b}.
\]
By Lemmas \ref{lem:zero-bribe-worst-belief} and \ref{lem:a2-lowest-highest},
the worst-case belief for player 1 is $\delta_{\underline{v}_{1}}$
and we can focus on $\max_{b}\pi_{1}(\underline{v}_{1}|b,\delta_{\underline{v}_{1}})$.
Player 2 accepts if and only if $\bar{b}\ge U_{2}(\bar{v}_{2}|\bar{\sigma}_{2})=\bar{v}_{2}c_{1,\bar{\sigma}_{2}}$.
Combining these, peace is implementable if and only if: 
\begin{equation}
\bar{v}_{2}c_{1,\bar{\sigma}_{2}}+\max_{b}\pi_{1}(\underline{v}_{1}|b,\delta_{\underline{v}_{1}})\le\underline{v}_{1}.\label{eq:impl-cond-combined}
\end{equation}

Case 1: $\bar{v}_{2}\le\underline{v}_{1}$. Any $b>0$ is accepted
with probability 1 (Lemma \ref{lem:a2lowBar-exis-uniq}). For $b=0$,
full acceptance yields $\pi_{1}(\underline{v}_{1}|0,\delta_{\underline{v}_{1}})=\underline{v}_{1}$.
Thus, $\max_{b}\pi_{1}=\underline{v}_{1}$. Equation (\ref{eq:impl-cond-combined})
becomes $\bar{v}_{2}c_{1,\bar{\sigma}_{2}}+\underline{v}_{1}\le\underline{v}_{1}$,
implying $c_{1,\bar{\sigma}_{2}}=0$.

Case 2: $\bar{v}_{2}>\underline{v}_{1}$. By Lemma \ref{lem:zero-bribe-payoff-is-not-max},
the maximum deviation payoff occurs at some $b>0$. If $b\ge\bar{v}_{2}-\underline{v}_{1}$,
the bribe is fully accepted, and $\pi_{1}(\underline{v}_{1}|b,\delta_{\underline{v}_{1}})=\underline{v}_{1}-b$.
The maximum in this range is at $b=\bar{v}_{2}-\underline{v}_{1}$,
yielding $2\underline{v}_{1}-\bar{v}_{2}$. If $0<b<\bar{v}_{2}-\underline{v}_{1}$,
the rejection set is non-empty. By Lemma \ref{lem:a2lowBar-exis-uniq},
$a_{2,\underline{\sigma}}(b)$ is unique and strictly increasing for
$b\ge b_{\min}$. The payoff is $\pi_{1}(\underline{v}_{1}|b,\delta_{\underline{v}_{1}})=F_{2}(a_{2,\underline{\sigma}}(b))(\underline{v}_{1}-b)$.
Using the inverse relationship between $b$ and $a_{2}$ from (\ref{interpret-as-bidding-1st-price}),
the maximization problem becomes: 
\[
\tilde{\pi}_{1}^{*}(\underline{v}_{1}):=\max_{a_{2}\in[\underline{v}_{2},\bar{v}_{2}]}F_{2}(a_{2})\left(\underline{v}_{1}-a_{2}(1-\underline{v}_{1}\mathcal{I}_{2}(a_{2}))\right).
\]
Substituting this into (\ref{eq:impl-cond-combined}) yields condition
(\ref{imple-thm}).

Robustness: By Lemma \ref{slope-lower-than-1}, if type $v_{1}$ profits
from deviation, so does $\underline{v}_{1}$. Thus, the belief $\delta_{\underline{v}_{1}}$
survives the D1 criterion. Similarly, since player 2's rejection payoff
is increasing in type, the belief $\delta_{\bar{v}_{2}}$ survives
D1 upon rejection. Hence, condition (\ref{imple-thm}) characterizes
robust equilibria.

\section{Proofs for the requesting model}

For any off-path request $r$ and any payoff-relevant consistent reply,
the deviation payoff of type $v_{1}$ is increasing in $v_{1}$, while
the equilibrium payoff in a peaceful requesting equilibrium is the
constant $\bar{r}$ independent of $v_{1}$. Hence D1 selects $\bar{v}_{1}$.
We summarize this auxiliary observation below for the proofs in this
section. 
\begin{lem}
\label{lem:D1-selects-v1bar}In any robust peaceful equilibrium of
the requesting model, the D1 criterion selects the posterior belief
$\delta_{\bar{v}_{1}}$ after every off-path request. 
\end{lem}

\subsection{Proof of Lemma \ref{lem-request-non-exist-bound} \protect}\label{proof-request-non-exist-bound}

Consider first the case $\bar{v}_{1}>2\underline{v}_{2}$. In any
peaceful equilibrium, $\bar{r}\le\underline{v}_{2}$ and player 1's
payoff is $\bar{r}$. Suppose type $\bar{v}_{1}$ deviates to an off-path
request $r\in(\underline{v}_{2},\bar{v}_{1}/2)$. Upon receiving the
off-path request $r$, no rejecting type $v_{2}$ would bid more than
$r$ in the continuation auction. So for any belief $\tilde{F}_{1}$
and any $\sigma$, $x_{\sigma}\le r$. So $\bar{v}_{1}-x_{\sigma}\ge\bar{v}_{1}-r>\bar{v}_{1}-\bar{v}_{1}/2=\bar{v}_{1}/2>\underline{v}_{2}\ge\bar{r}$.
Since $r>\bar{r}$, for any belief $\tilde{F}_{1}$ and any $\sigma$,
type $\bar{v}_{1}$'s expected payoff $\pi_{1}(\bar{v}_{1}|r,\tilde{F}_{1})>\bar{r}$
by (\ref{eq:pi-v1bar-r}). Therefore, such an off-path $r\in(\underline{v}_{2},\bar{v}_{1}/2)$
is always a profitable deviation for type $\bar{v}_{1}$. So in this
case, there is no peaceful equilibrium.

So there exists a peaceful equilibrium only if $\bar{v}_{1}\le2\underline{v}_{2}$.
By similar arguments to those above, it follows that $\bar{r}\ge\bar{v}_{1}/2$
in any peaceful equilibrium. To see this, suppose to the contrary,
$\bar{r}<\bar{v}_{1}/2$. Consider type $\bar{v}_{1}$'s deviation
to $r=\bar{v}_{1}/2$. Again, in the continuation auction, $x_{\sigma}\le r=\bar{v}_{1}/2$
for any beliefs. So $\pi_{1}(\bar{v}_{1}|r=\bar{v}_{1}/2,\tilde{F}_{1})\ge\bar{v}_{1}/2>\bar{r}$
and thus $r=\bar{v}_{1}/2$ is a profitable off-path deviation for
type $\bar{v}_{1}$ for any beliefs.

\subsection{Proof of Lemma \ref{lem:request-no-equil-v1bar-too-low} \protect}\label{proof-request-no-equil-v1bar-too-low}

By Lemma \ref{lem:D1-selects-v1bar}, in any robust peaceful equilibrium
D1 selects the belief $\delta_{\bar{v}_{1}}$ after every off-path
request. We therefore focus on this belief below.

We first show that any possible peaceful request cannot equal or exceed
$\bar{v}_{1}$. To see this, suppose to the contrary that such an
equilibrium exists and thus in the equilibrium player 2's payoff is
no higher than $v_{2}-\bar{v}_{1}$. Consider player 2's deviation,
namely rejection of $\bar{r}$. For any given $\tilde{F}_{2}$, in
any BNE $\sigma_{2}$ of $\mathcal{G}(F_{1},\tilde{F}_{2})$, because
$\bar{v}_{1}>\underline{v}_{1}$, we must have that $x_{\sigma_{2}}<\bar{v}_{1}$
by Lemma \ref{lem:highest-type-earns-positively}. It follows then
that any type $v_{2}\ge\bar{v}_{1}$ at least can bid $x_{\sigma_{2}}$
to secure a payoff strictly higher than $v_{2}-\bar{v}_{1}$ (the
equilibrium expected payoff would only be higher). So rejection of
such a request is a profitable deviation and thus a contradiction.
Thus below we consider $\bar{r}<\bar{v}_{1}$.

We next show that if $\bar{v}_{1}\le\underline{v}_{2}$ and $\bar{r}<\bar{v}_{1}$,
then type $\bar{v}_{1}$ can always deviate to some $r\in(\bar{r},\bar{v}_{1})$
which is accepted by all types of player 2 and thus a profitable deviation.
To see this, suppose $\bar{r}<\bar{v}_{1}$. Upon receiving an off-path
request $r\in(\bar{r},\bar{v}_{1})$, with belief $\delta_{\bar{v}_{1}}$
the rejection set of player 2 is $[\underline{v}_{2},\alpha_{2,\bar{\sigma}}(r)]$
for some $\alpha_{2,\bar{\sigma}}(r)>\underline{v}_{2}$. In any BNE
$\bar{\sigma}$ of $\mathcal{G}(\delta_{\bar{v}_{1}},\Psi_{2}(v_{2}|\alpha_{2,\bar{\sigma}}(r)))$,
$x_{\bar{\sigma}}=\bar{v}_{1}$ by Lemma \ref{lem:one-sided-complete-info}
because $\bar{v}_{1}\le\underline{v}_{2}$ implies $c_{2,\bar{\sigma}}=0$.\footnote{If $c_{2,\bar{\sigma}}>0$, then it means that there is a positive
measure of type $v_{2}$ earning zero payoff by bidding zero in $\bar{\sigma}$;
but any type $v_{2}>\bar{v}_{1}$ can secure a positive payoff by
bidding $\bar{v}_{1}$.} But then this violates the off-path consistency requirement on player
2 because the payoff of type $\alpha_{2,\bar{\sigma}}(r)$ is $\alpha_{2,\bar{\sigma}}(r)-x_{\bar{\sigma}}=\alpha_{2,\bar{\sigma}}(r)-\bar{v}_{1}$
and thus the fact that $r<\bar{v}_{1}$ implies that it is better
for type $\alpha_{2,\bar{\sigma}}(r)$ to accept $r$. Therefore,
the consistent reply of player 2 is to accept $r$ for all types $v_{2}$.
But full acceptance then implies that this off-path $r\in(\bar{r},\bar{v}_{1})$
is a profitable deviation for type $\bar{v}_{1}$. This completes
the proof.

\subsection{Proof of Theorem \ref{thm-requesting} \protect }\label{sec:Proof-request-lmm-bidder-1}

By Lemma \ref{lem:D1-selects-v1bar}, in any robust peaceful equilibrium
D1 selects the belief $\delta_{\bar{v}_{1}}$ after every off-path
request. We therefore focus on this belief below.

We first show that the equilibrium request must be $\bar{r}=\underline{v}_{2}$.
Suppose $\bar{r}<\underline{v}_{2}$. For any off-path $r\in(\bar{r},\underline{v}_{2})$,
player 2's consistent reply is full acceptance. If the rejection set
were non-empty, type $\underline{v}_{2}$ (the lowest type) would
be in it. In the continuation auction $\mathcal{G}(\delta_{\bar{v}_{1}},\Psi_{2}(v_{2}|\alpha_{2,\bar{\sigma}}(r)))$,
if $c_{1,\sigma}>0$, then $x_{\sigma}=\bar{v}_{1}>r$, so type $\alpha_{2,\bar{\sigma}}(r)$
prefers acceptance. If $c_{1,\sigma}=0$, type $\underline{v}_{2}$
earns 0, which is less than $r$. Thus, full acceptance is the only
consistent reply. Since full acceptance yields $\bar{v}_{1}>\bar{r}$,
deviating to such $r$ is profitable for player 1. Hence, $\bar{r}$
cannot be less than $\underline{v}_{2}$. Since $\bar{r}\le\underline{v}_{2}$
by definition, we have $\bar{r}=\underline{v}_{2}$.

With $\bar{r}=\underline{v}_{2}$, deviations to $r<\underline{v}_{2}$
are not profitable (full acceptance by the arguments above). We focus
on $r>\underline{v}_{2}$ and clearly for these off-path requests
the rejection set is non-empty, i.e., $\alpha_{2,\bar{\sigma}}(r)>\underline{v}_{2}$.

Recall that $x_{\bar{\sigma}^{*}}$ is the highest bid in any BNE
$\bar{\sigma}^{*}$ of $\mathcal{G}(\delta_{\bar{v}_{1}},F_{2})$.
So when $\bar{v}_{1}>\underline{v}_{2}$, it must be that $x_{\bar{\sigma}^{*}}\ge\underline{v}_{2}$.
This is because, if $x_{\bar{\sigma}^{*}}<\underline{v}_{2}$, then
both type $\bar{v}_{1}$ and type $\underline{v}_{2}$ win with positive
expected payoffs (and thus probabilities) and thus $c_{1,\bar{\sigma}^{*}},c_{2,\bar{\sigma}^{*}}>0$,
which is impossible.

Consider first an $r>x_{\bar{\sigma}^{*}}$. We show that for such
an $r$ the consistent reply of player 2 is full rejection and thus
it is not profitable for type $\bar{v}_{1}$ to deviate to such an
$r$. To see this, suppose that the rejection set is $[\underline{v}_{2},\alpha_{2,\bar{\sigma}}(r)]$
and $\alpha_{2,\bar{\sigma}}(r)<\bar{v}_{2}$. Then since $\alpha_{2,\bar{\sigma}}(r)$
is in the interior of the support of $F_{2}$, the indifference condition
implies that in any BNE $\bar{\sigma}$ of the continuation auction
$\mathcal{G}(\delta_{\bar{v}_{1}},\Psi_{2}(v_{2}|\alpha_{2,\bar{\sigma}}(r)))$,
$x_{\bar{\sigma}}=r>x_{\bar{\sigma}^{*}}$. From Lemma \ref{lem:one-sided-complete-info},
$x_{\bar{\sigma}}=\bar{v}_{1}(1-c_{2,\bar{\sigma}})$. Because $x_{\bar{\sigma}^{*}}=\bar{v}_{1}(1-c_{2,\bar{\sigma}^{*}})$,
we have $c_{2,\bar{\sigma}^{*}}>c_{2,\bar{\sigma}}\ge0$. Because
$c_{2,\bar{\sigma}^{*}}>0$ and by definition of $c_{2,\bar{\sigma}^{*}}$
we have $\bar{v}_{1}\int_{c_{2,\bar{\sigma}^{*}}}^{1}\left(F_{2}^{-1}(s)\right)^{-1}ds=1$,
it follows that 
\[
\bar{v}_{1}\int_{c_{2,\bar{\sigma}}}^{1}\left(F_{2}^{-1}(s)\right)^{-1}ds>1.
\]
Moreover, since $\Psi_{2}^{-1}(s|\alpha_{2,\bar{\sigma}}(r))=F_{2}^{-1}(sF_{2}(\alpha_{2,\bar{\sigma}}(r)))\le F_{2}^{-1}(s)$
for every $s\in[0,1]$, we have\\ $1/\Psi_{2}^{-1}(s|\alpha_{2,\bar{\sigma}}(r))\ge1/F_{2}^{-1}(s)$
and hence 
\[
\bar{v}_{1}\int_{c_{2,\bar{\sigma}}}^{1}\left(\Psi_{2}^{-1}(s|\alpha_{2,\bar{\sigma}}(r))\right)^{-1}ds\ge\bar{v}_{1}\int_{c_{2,\bar{\sigma}}}^{1}\left(F_{2}^{-1}(s)\right)^{-1}ds>1.
\]
Applying the boundary condition (\ref{eq:boundary-cond-one-sided})
to $\mathcal{G}(\delta_{\bar{v}_{1}},\Psi_{2}(\cdot|\alpha_{2,\bar{\sigma}}(r)))$
then gives 
\[
c_{1,\bar{\sigma}}=1-\bar{v}_{1}\int_{c_{2,\bar{\sigma}}}^{1}\left(\Psi_{2}^{-1}(s|\alpha_{2,\bar{\sigma}}(r))\right)^{-1}ds<0,
\]
which is impossible. Therefore, for any $r>x_{\bar{\sigma}^{*}}$,
the consistent reply of player 2 is full rejection, i.e., $\alpha_{2,\bar{\sigma}}(r)=\bar{v}_{2}$.
From above, in any BNE $\bar{\sigma}^{*}$ of $\mathcal{G}(\delta_{\bar{v}_{1}},F_{2})$,
$x_{\bar{\sigma}^{*}}\ge\underline{v}_{2}$. Because $\bar{v}_{1}\le2\underline{v}_{2}$,
type $\bar{v}_{1}$'s payoff is $\bar{v}_{1}-x_{\bar{\sigma}^{*}}\le\underline{v}_{2}$.
So with full rejection, it is not profitable for type $\bar{v}_{1}$
to deviate to such an $r$.

Next consider an $r\in(\underline{v}_{2},x_{\bar{\sigma}^{*}})$.
In this case full rejection is inconsistent because type $\bar{v}_{2}$
earns $\bar{v}_{2}-x_{\bar{\sigma}^{*}}<\bar{v}_{2}-r<r$. Thus, there
is a unique partial rejection threshold $\alpha_{2,\bar{\sigma}}(r)<\bar{v}_{2}$
given by (\ref{eq:request-a2sig-below-1}).\footnote{Since $\bar{v}_{1}\ge x_{\bar{\sigma}^{*}}>r=x_{\bar{\sigma}}$ and
thus type $\bar{v}_{1}$ earns a positive expected payoff, $c_{2,\bar{\sigma}}>0$
and $c_{1,\bar{\sigma}}=0$. Applying $c_{1,\bar{\sigma}}=0$ to (\ref{eq:boundary-cond-one-sided})
for $\mathcal{G}(\delta_{\bar{v}_{1}},\Psi_{2}(v_{2}|\alpha_{2,\bar{\sigma}}(r)))$,
it is clear that $\alpha_{2,\bar{\sigma}}(r)$ is given by (\ref{eq:request-a2sig-below-1}).} With $x_{\sigma}=r$, player 1's payoff is $F_{2}(\alpha_{2,\bar{\sigma}}(r))(\bar{v}_{1}-r)+(1-F_{2}(\alpha_{2,\bar{\sigma}}(r)))r$.

We can then formulate the maximization problem in (\ref{eq:p1-v1bar-r-delta-v1bar})
with the compact choice set including $\underline{v}_{2}$ and $x_{\bar{\sigma}^{*}}$.
So in any robust peaceful equilibrium, it is not profitable for any
type $v_{1}$ to deviate to any off-path request if and only if 
\[
\underline{v}_{2}\ge F_{2}(\alpha_{2,\bar{\sigma}}(r^{*}))(\bar{v}_{1}-r^{*})+(1-F_{2}(\alpha_{2,\bar{\sigma}}(r^{*})))r^{*}.
\]

Now we consider player 2's rejection of the on-path request. In any
robust peaceful equilibrium with $\bar{r}=\underline{v}_{2}$, if
the equilibrium request $\bar{r}$ is rejected, the only reasonable
belief about $v_{2}$ is that $v_{2}=\underline{v}_{2}$ because for
any belief, the expected payoff of player 2 is non-decreasing in $v_{2}$
with derivative no greater than one wherever it exists. Thus it is
not profitable for any type $v_{2}$ to reject $\bar{r}$ if and only
if $U_{2}(\underline{v}_{2}|\underline{\sigma}_{2})=0$, namely 
\[
\underline{v}_{2}c_{1,\underline{\sigma}_{2}}=0.
\]
Since $\underline{v}_{2}$ is required to be positive, $c_{1,\underline{\sigma}_{2}}=0$
is required.

Aggregating all the results above, the proof is complete.

\subsection{Proof of Lemma \ref{lem:alpha2-sigma-ranking} \protect }\label{Proof-alpha2-sigma-ranking}

We first show that $\alpha_{2,\underline{\sigma}}\ge\alpha_{2,\sigma}$.
To show this, suppose to the contrary, $\alpha_{2,\underline{\sigma}}<\alpha_{2,\sigma}$.

Given an off-path request $r$, consider the BNE $\underline{\sigma}$
of $\mathcal{G}(\delta_{\underline{v}_{1}},\Psi_{2}(v_{2}|\alpha_{2,\underline{\sigma}}))$
induced by belief $\delta_{\underline{v}_{1}}$. Also consider some
BNE $\sigma$ of $\mathcal{G}(\tilde{F}_{1},\Psi_{2}(v_{2}|\alpha_{2,\sigma}))$
induced by some belief $\tilde{F}_{1}$. If $\alpha_{2,\sigma}<\bar{v}_{2}$,
then $x_{\sigma}=x_{\underline{\sigma}}=r$. If $\alpha_{2,\sigma}=\bar{v}_{2}$,
then $x_{\sigma}\le x_{\underline{\sigma}}=r$. Hence, $\alpha_{2,\underline{\sigma}}<\alpha_{2,\sigma}$
implies 
\[
x_{\sigma}\le x_{\underline{\sigma}}.
\]

For convenience, we abuse notation a little by denoting the type distribution
functions $\tilde{F}_{i}$ and bid distribution functions $\tilde{H_{i}}$
in the BNE $\underline{\sigma}$ by $F_{i,\underline{\sigma}}$ and
$H_{i,\underline{\sigma}}$, while in the BNE $\sigma$ by $F_{i,\sigma}$
and $H_{i,\sigma}$ for any generic $\tilde{F}_{1}\ne\delta_{\underline{v}_{1}}$.

From Lemma \ref{general-all-pay}, in $\underline{\sigma}$, 
\begin{gather*}
H_{1,\underline{\sigma}}'(\beta)=\frac{1}{F_{2,\underline{\sigma}}^{-1}(H_{2,\underline{\sigma}}(\beta))},\;H_{2,\underline{\sigma}}'(\beta)=\frac{1}{F_{1,\underline{\sigma}}^{-1}(H_{1,\underline{\sigma}}(\beta))}.
\end{gather*}
Similarly, in $\sigma$, 
\begin{gather*}
H_{1,\sigma}'(\beta)=\frac{1}{F_{2,\sigma}^{-1}(H_{2,\sigma}(\beta))},\;H_{2,\sigma}'(\beta)=\frac{1}{F_{1,\sigma}^{-1}(H_{1,\sigma}(\beta))}.
\end{gather*}

Observe that $F_{1,\underline{\sigma}}^{-1}(H)=\underline{v}_{1}\le F_{1,\sigma}^{-1}\left(\tilde{H}\right)$
for any $H$ and $\tilde{H}$, and $\underline{v}_{1}<F_{1,\sigma}^{-1}\left(\tilde{H}\right)$
for a positive mass of $\tilde{H}$, which implies $H_{2,\underline{\sigma}}'(\beta)\ge H_{2,\sigma}'(\beta)$
for each $\beta$ and $H_{2,\underline{\sigma}}'(\beta)>H_{2,\sigma}'(\beta)$
for a positive mass of $\beta$. Because $x_{\sigma}\le x_{\underline{\sigma}}$,
\[
H_{2,\underline{\sigma}}(\beta)=1-\int_{\beta}^{x_{\underline{\sigma}}}H_{2,\underline{\sigma}}'(x)dx\le1-\int_{\beta}^{x_{\sigma}}H_{2,\sigma}'(x)dx=H_{2,\sigma}(\beta)
\]
for any $\beta$, and in particular, 
\begin{equation}
c_{2,\underline{\sigma}}=H_{2,\underline{\sigma}}(0)<H_{2,\sigma}(0)=c_{2,\sigma}.\label{eq:c2-equal-requesting}
\end{equation}
It then follows that $c_{2,\sigma}>0$. 

It also follows from $\alpha_{2,\underline{\sigma}}<\alpha_{2,\sigma}$
that $F_{2,\underline{\sigma}}(v_{2})>F_{2,\sigma}(v_{2})$ for $v_{2}\notin\{\underline{v}_{2},\alpha_{2,\sigma}\}$,
and thus $F_{2,\underline{\sigma}}^{-1}(H)<F_{2,\sigma}^{-1}(H)$
for each $H\ne0$.\footnote{$F_{2,\sigma}(v_{2})$ is the conditional distribution $F_{2}(v_{2}|v_{2}\le\alpha_{2,\sigma})$.
So stochastic dominance is implied.} Because $F_{2,\underline{\sigma}}^{-1}(H)$ and $F_{2,\sigma}^{-1}(H)$
are both increasing and $H_{2,\underline{\sigma}}(\beta)\le H_{2,\sigma}(\beta)$
for any $\beta$, $F_{2,\underline{\sigma}}^{-1}(H_{2,\underline{\sigma}}(\beta))<F_{2,\sigma}^{-1}(H_{2,\sigma}(\beta))$.
So $H_{1,\underline{\sigma}}'(\beta)>H_{1,\sigma}'(\beta)$ for any
$\beta>0$. Since $x_{\sigma}\le x_{\underline{\sigma}}$, we have
\[
c_{1,\underline{\sigma}}=1-\int_{0}^{x_{\underline{\sigma}}}H_{1,\underline{\sigma}}'(\beta)d\beta<1-\int_{0}^{x_{\sigma}}H_{1,\sigma}'(\beta)d\beta=c_{1,\sigma}.
\]
It then follows that $c_{1,\sigma}>0$.

Therefore, the fact that $c_{1,\sigma},c_{2,\sigma}>0$ implies $c_{1,\sigma}c_{2,\sigma}>0$,
a contradiction to the equilibrium requirement $c_{1,\sigma}c_{2,\sigma}=0$.
The supposition is false and we have $\alpha_{2,\underline{\sigma}}\ge\alpha_{2,\sigma}$,
which also implies $x_{\underline{\sigma}}\le x_{\sigma}$.

The proof for $\alpha_{2,\sigma}\ge\alpha_{2,\bar{\sigma}}$ and $x_{\sigma}\le x_{\bar{\sigma}}$
can be done in the same spirit and thus is omitted.

\subsection{Proof of Theorem \ref{thm:request-security} \protect}\label{subsec:Proof-of-request-security}

We assume throughout that a robust peaceful equilibrium exists, so
$\bar{r}=\underline{v}_{2}$ and $\bar{v}_{1}\le2\underline{v}_{2}$.
Security must rule out two kinds of profitable deviations: (a) player~2's
rejection of the on-path request $\bar{r}$, and (b) player~1's deviation
to an off-path request $r\ne\bar{r}$. The proof is organized in three
steps.

\emph{Step 1: Player 2's rejection of the on-path request.} The expected
payoff from rejecting $\bar{r}$ is increasing in $v_{2}$ with slope
$\le1$, while the acceptance payoff increases at rate 1. Thus, it
suffices to check type $\underline{v}_{2}$. Under the belief maximizing
player 2's continuation payoff ($\tilde{F}_{2}=\delta_{\bar{v}_{2}}$),
$U_{2}(\underline{v}_{2}\mid\bar{\sigma}_{2})=\underline{v}_{2}c_{1,\bar{\sigma}_{2}}$.
Security requires this to be $\le\underline{v}_{2}$, which simplifies
to $c_{1,\bar{\sigma}_{2}}=0$.

\emph{Step 2: Player 1's deviation to off-path requests when}\textbf{
$\underline{v}_{1}>\underline{v}_{2}$.} We show that in this case,
no profitable deviation exists for type $\bar{v}_{1}$ (and hence
for any type) beyond the condition $c_{1,\bar{\sigma}_{2}}=0$ already
established.

\noindent\emph{Case 2a: $r<\underline{v}_{2}$.} We claim that for
any belief $\tilde{F}_{1}$, the consistent reply to $r$ is full
acceptance. By Remark~\ref{rem:request-full-rej-acc}(iii), it suffices
to verify this for belief $\delta_{\underline{v}_{1}}$. Suppose,
to the contrary, that rejection occurs with positive probability under
some BNE $\underline{\sigma}$ of the induced continuation auction
$\mathcal{G}(\delta_{\underline{v}_{1}},\Psi_{2}(v_{2}\mid\alpha_{2,\underline{\sigma}}(r)))$.
If $c_{1,\underline{\sigma}}>0$, then type $\underline{v}_{1}$ earns
zero payoff and $x_{\underline{\sigma}}=\underline{v}_{1}$; but then
type $\alpha_{2,\underline{\sigma}}(r)$ would strictly prefer accepting
$r$ because $r<\underline{v}_{2}<\underline{v}_{1}=x_{\underline{\sigma}}$.
If $c_{1,\underline{\sigma}}=0$, then type $\underline{v}_{2}$ earns
zero from rejection but strictly positive payoff $\underline{v}_{2}-r>0$
from acceptance, a contradiction. Thus the reply is full acceptance
for belief $\delta_{\underline{v}_{1}}$ and hence for all beliefs.
No profitable deviation exists in this case.

\emph{Case 2b: $r>\underline{v}_{2}$.} Such a request must lead to
a positive probability of rejection. We distinguish full and partial
rejection.

If the consistent reply is \emph{full rejection} for some belief,
then by Remark~\ref{rem:request-full-rej-acc}(ii) it is also full
rejection for belief $\delta_{\underline{v}_{1}}$. By Lemma~\ref{lem:alpha2-sigma-ranking},
$x_{\underline{\sigma}}\le x_{\sigma}$ for any belief, so type $\bar{v}_{1}$'s
payoff $\bar{v}_{1}-x_{\sigma}$ is maximized under belief $\delta_{\underline{v}_{1}}$.
Let $x_{\underline{\sigma}^{*}}$ be the highest bid in $\mathcal{G}(\delta_{\underline{v}_{1}},F_{2})$.
Since $\underline{v}_{1}>\underline{v}_{2}$, we must have $x_{\underline{\sigma}^{*}}\ge\underline{v}_{2}$;
otherwise both types $\underline{v}_{1}$ and $\underline{v}_{2}$
would earn positive payoffs, violating $c_{1,\underline{\sigma}^{*}}c_{2,\underline{\sigma}^{*}}=0$.
Because $\bar{v}_{1}\le2\underline{v}_{2}$, type $\bar{v}_{1}$'s
payoff satisfies $\bar{v}_{1}-x_{\underline{\sigma}^{*}}\le\underline{v}_{2}=\bar{r}$.
The deviation is not profitable in this case.

If the consistent reply is \emph{partial rejection}, then type $\bar{v}_{1}$'s
expected payoff is 
\[
F_{2}(\alpha_{2,\sigma}(r))(\bar{v}_{1}-r)+(1-F_{2}(\alpha_{2,\sigma}(r)))r,
\]
a convex combination of $\bar{v}_{1}-r$ and $r$. Since $r>\underline{v}_{2}\ge\bar{v}_{1}/2$,
we have $\bar{v}_{1}-r<r$. So, given any $r$, the payoff is decreasing
in $\alpha_{2,\sigma}(r)$. By Lemma~\ref{lem:alpha2-sigma-ranking},
the maximum is attained at $\alpha_{2,\bar{\sigma}}(r)$ under belief
$\delta_{\bar{v}_{1}}$. Under D1, this is the selected belief after
an off-path request, and robustness of the peaceful equilibrium already
rules out a profitable deviation for type $\bar{v}_{1}$ under that
belief.

\emph{Step 3: Player 1's deviation to off-path requests when}\textbf{
$\underline{v}_{1}\le\underline{v}_{2}$.} We first establish the
necessity of $\bar{v}_{1}-\underline{v}_{1}\le\underline{v}_{2}$,
then show that together with $c_{1,\bar{\sigma}_{2}}=0$ it is sufficient.

\noindent\emph{Necessity.} Player 1 can deviate to a sufficiently
high request (e.g., $r>\bar{v}_{1}$), which triggers full rejection
under any belief. Under the most favorable belief for player~1, namely
$\delta_{\underline{v}_{1}}$, the continuation auction is $\mathcal{G}(\delta_{\underline{v}_{1}},F_{2})$.
Its highest bid equals $\underline{v}_{1}$ (if $x_{\underline{\sigma}}<\underline{v}_{1}\le\underline{v}_{2}$,
then both player 1's and player 2's lowest types earn positive payoffs,
violating $c_{1,\underline{\sigma}}c_{2,\underline{\sigma}}=0$).
Type $\bar{v}_{1}$ then earns $\bar{v}_{1}-\underline{v}_{1}$, and
security requires this not to exceed the equilibrium payoff $\underline{v}_{2}$.

\emph{Sufficiency.} We verify that $\bar{v}_{1}-\underline{v}_{1}\le\underline{v}_{2}$
rules out all profitable deviations for type $\bar{v}_{1}$. It is
useful to note the following characterization of player~2's reply
under belief $\delta_{\underline{v}_{1}}$: 
\begin{itemize}
\item If $r<\underline{v}_{1}$, full acceptance by the same arguments in
Case 2a. 
\item If $r>\underline{v}_{1}$, full rejection. Indeed, any type $v_{2}$
can bid $\underline{v}_{1}$ and earn $v_{2}-\underline{v}_{1}>v_{2}-r$. 
\item If $r=\underline{v}_{1}$, player~2 is indifferent between acceptance
and rejection. 
\end{itemize}
By Remark~\ref{rem:request-full-rej-acc}(iii), full acceptance under
belief $\delta_{\underline{v}_{1}}$ implies full acceptance under
any belief. Hence for $r<\underline{v}_{1}$, no profitable deviation
exists.

We now consider $r\ge\underline{v}_{1}$.

\emph{Full rejection.} By Remark~\ref{rem:request-full-rej-acc}(ii),
if full rejection occurs for some belief, it also occurs for $\delta_{\underline{v}_{1}}$.
By Lemma~\ref{lem:alpha2-sigma-ranking}, $x_{\underline{\sigma}}\le x_{\sigma}$,
so type $\bar{v}_{1}$'s payoff $\bar{v}_{1}-x_{\sigma}$ is maximized
under belief $\delta_{\underline{v}_{1}}$, where the highest bid
is $\underline{v}_{1}$ and the payoff is $\bar{v}_{1}-\underline{v}_{1}\le\underline{v}_{2}$.
The deviation is not profitable in this case.

\emph{Partial rejection.} Type $\bar{v}_{1}$'s expected payoff is
\begin{equation}
\pi_{1}(\bar{v}_{1}\mid r)=F_{2}(\alpha_{2,\sigma}(r))(\bar{v}_{1}-r)+(1-F_{2}(\alpha_{2,\sigma}(r)))r,\label{eq:partial-payoff}
\end{equation}
a convex combination of $\bar{v}_{1}-r$ and $r$. We distinguish
two sub-cases.

If $r>\underline{v}_{2}$, then $\bar{v}_{1}-r<\underline{v}_{2}<r$,
so the payoff in \eqref{eq:partial-payoff} is decreasing in $\alpha_{2,\sigma}(r)$.
The maximum is attained at $\alpha_{2,\bar{\sigma}}(r)$ under belief
$\delta_{\bar{v}_{1}}$ (Lemma~\ref{lem:alpha2-sigma-ranking}).
Robustness rules out a profitable deviation for type $\bar{v}_{1}$
under that belief. 

If $r\in[\underline{v}_{1},\underline{v}_{2}]$, then both endpoints
of the convex combination are bounded by $\underline{v}_{2}$. To
see this, note first that $r\le\underline{v}_{2}$ trivially. Second,
since $r\ge\underline{v}_{1}$, we have $\bar{v}_{1}-r\le\bar{v}_{1}-\underline{v}_{1}\le\underline{v}_{2}$.
Hence the payoff in \eqref{eq:partial-payoff} is at most $\underline{v}_{2}$
for any belief. The deviation is not profitable in this case. 

\emph{Conclusion.} We have shown that security requires $c_{1,\bar{\sigma}_{2}}=0$
(Step~1) and, when $\underline{v}_{1}\le\underline{v}_{2}$, additionally
$\bar{v}_{1}-\underline{v}_{1}\le\underline{v}_{2}$ (Step~3). When
$\underline{v}_{1}>\underline{v}_{2}$, the condition $\bar{v}_{1}-\underline{v}_{1}\le\underline{v}_{2}$
is automatically satisfied because $\bar{v}_{1}\le2\underline{v}_{2}$
(required by robustness) and $\underline{v}_{1}>\underline{v}_{2}$.
Thus the two conditions in \eqref{eq:request-security-player-2} are
necessary and sufficient.

\bibliographystyle{teURL-fixed}
\bibliography{bribing}

\clearpage

\begin{center}
{\Large\textbf{Supplemental Appendix}}{\Large\par}
\par\end{center}

 \setcounter{page}{1}
\setcounter{footnote}{0}
\setcounter{lem}{0}
\setcounter{section}{0}
\setcounter{thm}{0}

\section{Comparative statics for peace implementability in the bribing model}

\citet{Zheng2019} shows that in his mediation model, if one or both
players become stronger in the sense that their type distributions
become more FOSD and the support remains unchanged, then peace implementability
is preserved. The preservation result also holds in our model if player
1's type distribution becomes more FOSD. To see this, observe that
the first term in (\ref{imple-thm}) necessarily decreases if $F_{1}$
becomes more FOSD. On the other hand, the second term in (\ref{imple-thm})
is the expected payoff of type $\underline{v}_{1}$, which depends
only on the type distribution through the value of $\underline{v}_{1}$
and $F_{2}$.

The prospect for peace implementability becomes more subtle if $F_{2}$
becomes more FOSD (even if the support remains unchanged). Intuitively,
one might expect that the expected payoff of type $\underline{v}_{1}$
in the FPA should become smaller due to a stronger opponent. However,
when player 2 becomes stronger, her ``bidding'' strategy (namely
$a_{2,\sigma}(b)$) also changes. In particular, $\mathcal{I}_{2}(a_{2})$
in \eqref{eq:pi1-tilde} may increase for some $a_{2}$, causing $\tilde{\pi}_{1}^{*}(\underline{v}_{1})$
to rise and thereby rendering peace unimplementable.  Hence, FOSD
may not be sufficient for preserving peace implementability. However,
we show below that if player 2 becomes stronger in terms of hazard-rate
dominance, then peace implementability can be preserved.
\begin{lem}
\label{lem:I-phi-Phi-ranking}Consider two distribution functions
$F$ and $G$ with densities $f$ and $g$, respectively.  Suppose
$G$ dominates $F$ in terms of the hazard rate, i.e., 
\[
\lambda_{G}(x)=\frac{g(x)}{1-G(x)}\le\frac{f(x)}{1-F(x)}=\lambda_{F}(x).
\]
Let 
\[
\Phi(v|x)=\frac{F(v)-F(x)}{1-F(x)},\quad\phi(v|x)=\frac{G(v)-G(x)}{1-G(x)}.
\]
Denote the inverse functions by $\Phi^{-1}(\cdot|x)$ and $\phi^{-1}(\cdot|x)$
respectively. Let
\[
\mathcal{I}_{\Phi}(x):=\int_{0}^{1}\frac{1}{\Phi^{-1}(s|x)}ds,\quad\mathcal{I}_{\phi}(x):=\int_{0}^{1}\frac{1}{\phi^{-1}(s|x)}ds.
\]
Then for all $x$,
\[
\mathcal{I}_{\phi}(x)\le\mathcal{I}_{\Phi}(x).
\]
\end{lem}
\begin{proof}
It's a standard result that $1-F(v)=\exp\left(-\int_{0}^{v}\lambda_{F}(z)dz\right)$.
The survival function of the conditional distribution $\Phi(\cdot|x)$
is: 
\[
1-\Phi(v|x)=\frac{1-F(v)}{1-F(x)}=\exp\left(-\int_{x}^{v}\lambda_{F}(z)dz\right).
\]
Similarly, for $G$: 
\[
1-\phi(v|x)=\exp\left(-\int_{x}^{v}\lambda_{G}(z)dz\right).
\]
Since $\lambda_{F}(z)\ge\lambda_{G}(z)$ for all $z$, we have $\int_{x}^{v}\lambda_{F}(z)dz\ge\int_{x}^{v}\lambda_{G}(z)dz$
for any $v\ge x$. This implies: 
\[
1-\Phi(v|x)\le1-\phi(v|x)\implies\Phi(v|x)\ge\phi(v|x).
\]
Thus, $\phi$ first-order stochastically dominates $\Phi$, which
means $\phi^{-1}(s|x)\ge\Phi^{-1}(s|x)$ for all $s\in(0,1)$. Consequently:
\[
\frac{1}{\phi^{-1}(s|x)}\le\frac{1}{\Phi^{-1}(s|x)}.
\]
Integrating both sides over $[0,1]$ yields $\mathcal{I}_{\phi}(x)\le\mathcal{I}_{\Phi}(x)$.
\end{proof}
The key to the result above is that hazard-rate dominance is precisely
the condition needed to ensure that FOSD is preserved under truncation.

By Lemma \ref{lem:I-phi-Phi-ranking}, if player 2's type distribution
becomes more hazard-rate dominant, then in \eqref{eq:pi1-tilde},
for each $a_{2}$, $\mathcal{I}_{2}(a_{2})$ decreases. Hazard-rate
dominance also implies FOSD. Thus $F_{2}(a_{2})$ becomes smaller
for each $a_{2}$ and the objective function in \eqref{eq:pi1-tilde}
becomes smaller. Therefore the maximum can only become smaller and
peace implementability is preserved. 

We summarize these observations in the following result. 
\begin{thm}
Suppose peace is implementable for some prior ($F_{1},F_{2}$). Then
peace implementability is preserved if $F_{1}$ becomes more first\nobreakdash-order
stochastically dominant on the same support, or $F_{2}$ becomes more
hazard\nobreakdash-rate dominant.

\end{thm}

\section{Non-peaceful equilibria of the bribing model}

In this section we consider non-peaceful equilibria. In any equilibrium,
if an on-path bribe is accepted with probability one, then any other
on-path bribe must be rejected with positive probability.\footnote{\label{footnote-prob-one}If there is a different bribe $b'$ offered
by a different type $v_{1}'$ and accepted with probability one, then
the lower bribe is preferred by both types.} Lemma \ref{lem:c2sig-zero-p1-form} also implies that for any on-path
positive bribe $b$, the rejection set can be described by an interval
$[a_{2,\sigma}(b),\bar{v}_{2}]$ where $\sigma$ is a BNE of the continuation
auction.

We first rule out decreasing equilibria and non-monotonic equilibria. 
\begin{lem}
\label{lem:bribe-non-decreasing}In any equilibrium, the on-path bribing
function is non-decreasing. 
\end{lem}
\begin{proof}
Consider two on-path bribes $b_{h}$ and $b_{l}<b_{h}$. Let the infimum
type who offers $b_{h}$ be denoted by $v_{h}=\inf\{v_{1}:b(v_{1})=b_{h}\}$
and the supremum type who offers $b_{l}$ be denoted by $v_{l}=\sup\{v_{1}:b(v_{1})=b_{l}\}$.
Correspondingly, for $i=h,l$, let $a_{2,\sigma_{i}}$ be the lowest
rejecting type and $\sigma_{i}$ be a BNE of the continuation auction
following rejection of $b_{i}$.

Observe that if $b_{h}$ is rejected, then in the BNE of the induced
continuation auction $c_{2,\sigma_{h}}=0$. This is because rejection
of a positive bribe implies that the lowest rejecting type must earn
a positive payoff and thus $c_{2,\sigma_{h}}$ cannot be positive
(if it were, then $c_{1,\sigma_{h}}=0$ and thus her payoff $c_{1,\sigma_{h}}a_{2,\sigma_{h}}$
would be zero). Incentive compatibility requires that $\pi_{1}(v_{h}|b_{h})\ge\pi_{1}(v_{h}|b_{l})$,
i.e., 
\[
F_{2}(a_{2,\sigma_{h}})(v_{h}-b_{h})\ge F_{2}(a_{2,\sigma_{l}})(v_{h}-b_{l})+(1-F_{2}(a_{2,\sigma_{l}}))(\max_{\beta}\;H_{2,\sigma_{l}}(\beta)v_{h}-\beta).
\]
This implies 
\[
F_{2}(a_{2,\sigma_{h}})(v_{h}-b_{h})\ge F_{2}(a_{2,\sigma_{l}})(v_{h}-b_{l}).
\]
Because $b_{h}>b_{l}$, we must have $F_{2}(a_{2,\sigma_{h}})>F_{2}(a_{2,\sigma_{l}})$
or equivalently $a_{2,\sigma_{h}}>a_{2,\sigma_{l}}$.

In the equilibrium, if an on-path bribe $b$ induces a continuation
auction with a BNE $\sigma$, the expected payoff of type $v_{1}$
offering it is 
\[
\pi(v_{1},b)=F_{2}(a_{2,\sigma}(b))(v_{1}-b)+\left(1-F_{2}(a_{2,\sigma}(b))\right)(\max_{\beta}\;H_{2,\sigma}(\beta)v_{1}-\beta).
\]
Since $a_{2,\sigma}(b)$ is strictly increasing, it's differentiable
almost everywhere. So by the envelope theorem
\[
\frac{\partial}{\partial v_{1}}\pi(v_{1},b)=F_{2}(a_{2,\sigma}(b))+\left(1-F_{2}(a_{2,\sigma}(b))\right)H_{2,\sigma}(\beta(v_{1}))
\]
and
\[
\frac{\partial^{2}}{\partial v_{1}\partial b}\pi(v_{1},b)=f_{2}(a_{2,\sigma}(b))a_{2,\sigma}'(b)\left(1-H_{2,\sigma}(\beta(v_{1}))\right)>0.
\]
Thus, $\pi(v_{1},b)$ has the single-crossing property. By standard
arguments, the equilibrium bribing function is non-decreasing.
\end{proof}
We next consider equilibria with continuously increasing segments. 
\begin{lem}
\label{lem:no-separating-equil}There is no equilibrium in which the
on-path bribing  function is absolutely continuous and strictly increasing
over some open type interval. 
\end{lem}
\begin{proof}
We first note that at most one on-path bribe can be accepted with
probability one. Thus, we consider only on-path bribes $b(v_{1})$
that are rejected with positive probability. Suppose, for contradiction,
that there exists an open interval of types $v_{1}$ over which the
bribing function $b(v_{1})$ is absolutely continuous and strictly
increasing. Since the interval is open, $b(v_{1})>0$ for all $v_{1}$
in the interval.

Consider the continuation auction following rejection of a bribe $b(v_{1})$
from the interval. Player 2's best response is described by an interval
$[a_{2}(v_{1}),\bar{v}_{2}]$. Since by consistency type $a_{2}(v_{1})$
earns a positive expected payoff in any BNE $\sigma$ of the auction
(equal to $c_{1,\sigma}a_{2}(v_{1})$), it follows that $c_{1,\sigma}>0$.
So type $v_{1}$ earns zero payoff in the auction and the upper bound
of the common bidding interval is $v_{1}$.

Consider another type $\hat{v}_{1}$ of player 1 that mimics type
$v_{1}$ by offering bribe $b(v_{1})$. By Lemma \ref{secret-bidding},
if $\hat{v}_{1}>v_{1}$, type $\hat{v}_{1}$'s optimal choice is to
bid $v_{1}$ to win for sure and the payoff is thus $\hat{v}_{1}-v_{1}$.
Similarly if $\hat{v}_{1}<v_{1}$, then type $\hat{v}_{1}$'s optimal
choice is to bid zero and the payoff is zero.

Now consider a bribe $b(v_{1})$ from type $v_{1}$ and another bribe
$b(\hat{v}_{1})$ from type $\hat{v}_{1}$ in the neighborhood of
$v_{1}$.

If type $\hat{v}_{1}>v_{1}$ mimics type $v_{1}$ by offering $b(v_{1})$,
the expected payoff is 
\[
\pi_{1}(\hat{v}_{1},v_{1})=F_{2}(a_{2}(v_{1}))(\hat{v}_{1}-b(v_{1}))+(1-F_{2}(a_{2}(v_{1})))(\hat{v}_{1}-v_{1}).
\]
The incentive compatibility condition requires,\footnote{The condition $\hat{v}_{1}\in\arg\max_{\hat{v}_{1}}\;\pi_{1}(\hat{v}_{1},v_{1})$
yields 
\[
f_{2}(a_{2}(\hat{v}_{1}))(\hat{v}_{1}-b(\hat{v}_{1}))a_{2}'(\hat{v}_{1})-F_{2}(a_{2}(\hat{v}_{1}))b'(\hat{v}_{1})-(1-F_{2}(a_{2}(\hat{v}_{1})))=0.
\]
Replacing $\hat{v}_{1}$ by $v_{1}$ we obtain (\ref{ic-mimic-low}).} 
\begin{equation}
f_{2}(a_{2}(v_{1}))(v_{1}-b(v_{1}))a_{2}'(v_{1})-F_{2}(a_{2}(v_{1}))b'(v_{1})-(1-F_{2}(a_{2}(v_{1})))=0.\label{ic-mimic-low}
\end{equation}

If type $v_{1}<\hat{v}_{1}$ mimics type $\hat{v}_{1}$ by offering
a separating bribe $b(\hat{v}_{1})$, the expected payoff is 
\[
\pi_{1}(v_{1},\hat{v}_{1})=F_{2}(a_{2}(\hat{v}_{1}))(v_{1}-b(\hat{v}_{1})).
\]
The incentive compatibility condition requires 
\begin{equation}
f_{2}(a_{2}(v_{1}))(v_{1}-b(v_{1}))a_{2}'(v_{1})-F_{2}(a_{2}(v_{1}))b'(v_{1})=0.\label{ic-mimic-high}
\end{equation}
The conditions (\ref{ic-mimic-low}) and (\ref{ic-mimic-high}) together
imply 
\[
-(1-F_{2}(a_{2}(v_{1})))=0,
\]
which can be true if and only if $a_{2}(v_{1})=\bar{v}_{2}$ for all
$v_{1}$, and thus $b'(v_{1})=0$, a contradiction to the assumption
of the separating segment.
\end{proof}
The result above excludes regular separating segments, a finding reminiscent
of \citet{rachmilevitch2013}. In particular, it rules out differentiable
fully separating equilibria. The next theorem focuses on a finite-message
class. We call a non-peaceful equilibrium \emph{regular} if its on-path
bribing function is a finite step function (and hence has a finite
range and finitely many pooling regions). This restriction is not
merely technical; rather, it isolates the economically transparent
case in which the proposer uses finitely many settlement offers, excluding
pathological monotone strategies with infinitely many jumps or singular
components. The theorem should therefore be read as a characterization
within this regular class, not as an unrestricted characterization
of all possible monotone strategies.

We next examine other regular non-peaceful equilibria. The following
result shows that any such equilibrium has at most two on-path bribes.

\begin{thm}
\label{thm:no-semi-pooling}Suppose there exists a regular non-peaceful
equilibrium. Then either the equilibrium is fully pooling, or there
are exactly two on-path bribes, $b_{l}<b_{h}$. In the latter case,
lower types offer $b_{l}$, and higher types offer $b_{h}$. Moreover,
$b_{l}$ is rejected with positive probability, and $b_{h}$ is accepted
with probability one. 
\end{thm}
\begin{proof}
We show the result by showing that there exist no equilibria in which
there are two consecutive pooling bribes rejected with positive probability.
By Lemma \ref{lem:bribe-non-decreasing}, the regular (finite step
function) on-path bribing strategy is non-decreasing. Furthermore,
at most one on-path bribe can be accepted with probability one; if
two distinct bribes were accepted with probability one, the lower
bribe would be strictly preferred by all types offering the higher
bribe. Thus, if there are $n\ge2$ on-path bribes $b_{1}<\dots<b_{n}$,
only $b_{n}$ can be accepted with probability one, implying $b_{1},\dots,b_{n-1}$
are rejected with positive probability.

In this proof, for $a<b$, we write $[a,b\rangle$ to mean either
$[a,b]$ or $[a,b)$ and $\langle a,b]$ to mean either $[a,b]$ or
$(a,b]$.

Suppose there exists an equilibrium in which any type $v_{1}\in[v_{l},\hat{v}_{1}\rangle$
offers a bribe $b_{l}$ and any type $v_{1}\in\langle\hat{v}_{1},v_{h}]$
offers a bribe $b_{h}$ for some $v_{l}<\hat{v}_{1}<v_{h}$. By Lemma
\ref{lem:bribe-non-decreasing}, $b_{h}>b_{l}\ge0$. Suppose further
both $b_{h}$ and $b_{l}$ are rejected with positive probability.
Let the rejection set of $b_{i}$ be $[a_{2,\sigma^{i}}(b_{i}),\bar{v}_{2}]$
for $i=h,l$. And denote a BNE of the continuation auction following
rejection of $b_{i}$ by $\sigma^{i}$ and the highest bid by $x_{\sigma^{i}}$.

Clearly type $\hat{v}_{1}$ is indifferent between $b_{h}$ and $b_{l}$.
In $\sigma^{h}$, type $\hat{v}_{1}$ bids zero and the expected payoff
is $c_{2,\sigma^{h}}\hat{v}_{1}$, whereas in $\sigma^{l}$ type $\hat{v}_{1}$
bids $x_{\sigma^{l}}$ and the payoff is $\hat{v}_{1}-x_{\sigma^{l}}$.
Because by consistency type $a_{2,\sigma^{h}}$ must earn a positive
expected payoff in any BNE $\sigma$ of the auction, it follows that
$c_{1,\sigma^{h}}>0$ and $c_{2,\sigma^{h}}=0$. So the indifference
condition for type $\hat{v}_{1}$ is 
\[
F_{2}(a_{2,\sigma^{h}})(\hat{v}_{1}-b_{h})=F_{2}(a_{2,\sigma^{l}})(\hat{v}_{1}-b_{l})+(1-F_{2}(a_{2,\sigma^{l}}))(\hat{v}_{1}-x_{\sigma^{l}}).
\]
The fact $c_{1,\sigma^{h}}>0$ means that a positive measure of types
$v_{1}$ in the right neighborhood of $\hat{v}_{1}$ earn zero in
$\sigma^{h}$. It follows that the highest bid in $\sigma^{h}$, $x_{\sigma^{h}}>\hat{v}_{1}$.
So 
\begin{align*}
F_{2}(a_{2,\sigma^{h}})(\hat{v}_{1}-b_{h})+(1-F_{2}(a_{2,\sigma^{h}}))(\hat{v}_{1}-x_{\sigma^{h}})< & F_{2}(a_{2,\sigma^{l}})(\hat{v}_{1}-b_{l})+(1-F_{2}(a_{2,\sigma^{l}}))(\hat{v}_{1}-x_{\sigma^{l}})
\end{align*}
which is equivalent to 
\begin{equation}
F_{2}(a_{2,\sigma^{h}})b_{h}+(1-F_{2}(a_{2,\sigma^{h}}))x_{\sigma^{h}}>F_{2}(a_{2,\sigma^{l}})b_{l}+(1-F_{2}(a_{2,\sigma^{l}}))x_{\sigma^{l}}.\label{eq:semi-pooling-ineq}
\end{equation}

In the equilibrium, the expected payoff of type $v_{h}$ is 
\begin{align*}
\pi_{1}(v_{h}) & =F_{2}(a_{2,\sigma^{h}})(v_{h}-b_{h})+(1-F_{2}(a_{2,\sigma^{h}}))(v_{h}-x_{\sigma^{h}})\\
 & =v_{h}-[F_{2}(a_{2,\sigma^{h}})b_{h}+(1-F_{2}(a_{2,\sigma^{h}}))x_{\sigma^{h}}].
\end{align*}
Similarly, by deviating to $b_{l}$, the expected payoff of type $v_{h}$
is 
\[
\pi_{1}(v_{h}|b_{l})=v_{h}-[F_{2}(a_{2,\sigma^{l}})b_{l}+(1-F_{2}(a_{2,\sigma^{l}}))x_{\sigma^{l}}].
\]
So it follows from (\ref{eq:semi-pooling-ineq}) that $\pi_{1}(v_{h}|b_{l})>\pi_{1}(v_{h})$
and thus $b_{l}$ is a profitable deviation for type $v_{h}$, a contradiction.
\end{proof}

\end{document}